\documentclass[12pt]{article}
\usepackage{graphicx} 
\usepackage{epsfig}
\usepackage{amsmath, amsfonts, amsthm,amssymb}
\usepackage{latexsym}
\usepackage{xcolor}
\usepackage{mathrsfs}
\usepackage{natbib}
\usepackage{hyperref}
\usepackage{appendix}
\usepackage{dsfont}
\usepackage[left=3cm,top=4cm,right=3cm,bottom=3cm]{geometry}
\RequirePackage[caption=false]{subfig}
\usepackage{lscape}
\usepackage{dsfont}
\usepackage{rotating}
\usepackage{ulem}

\usepackage{mathabx}
\usepackage{calligra}

\newcommand{\de}{\mathrm{d}}
\DeclareMathOperator{\tf}{\mathfrak{t}}

\newcommand{\e}{\mathds{E}}
\newcommand{\R}{\mathds{R}}

\newcommand{\1}{\mathds{1}}

\newtheorem{lemma}{Lemma}

\graphicspath{{./figs/}}

\begin{document}

\title{\LARGE \bf{Inhomogeneous mark correlation functions for general  marked point processes} }
\maketitle
\begin{center}
{{\bf Mehdi Moradi$^{1}$}} and {\bf Matthias Eckardt$^{2}$} \\
\noindent $^{\text{1}}$ Department of Mathematics and Mathematical Statistics, Ume\r{a} University, Ume\r{a}, Sweden\\
\noindent $^{\text{2}}$ Chair of Statistics, Humboldt-Universit\"{a}t zu Berlin , Berlin, Germany\\

\end{center}
\begin{abstract}
Spatial phenomena in environmental and biological contexts often involve events that are unevenly distributed across space and carry attributes, whose associations/variations are space-dependent. 
In this paper, we introduce the class of inhomogeneous mark correlation functions, capturing mark associations/variations, while explicitly accounting for the spatial inhomogeneity of events.
The proposed functions are designed to quantify how, on average, marks vary or associate with one another as a function of pairwise spatial distances.
We develop nonparametric estimators and evaluate their performance through simulation studies covering a range of scenarios with mark association or variation, spanning from nonstationary point patterns without spatial interaction to those characterised by clustering tendencies.
Our simulations reveal the shortcomings of traditional methods in the presence of spatial inhomogeneity, underscoring the necessity of our approach. Furthermore, the results show that our estimators accurately identify both the positivity/negativity and effective spatial range for detected mark associations/variations.
The proposed inhomogeneous mark correlation functions are then applied to two distinct forest ecosystems:  Longleaf pine trees in southern Georgia, USA, marked by their diameter at breast height, and Scots pine trees in Pfynwald, Switzerland, marked by their height. Our findings reveal that the inhomogeneous mark correlation functions provide deeper and more detailed insights into tree growth patterns compared to traditional methods.
\end{abstract}


\maketitle

\section{Introduction}

In ecological and environmental sciences, capturing spatially structured relationships between individual organisms is vital for understanding how ecosystems function and evolve.
In particular, measuring the strength and nature of space-dependent associations/variations between tree attributes is fundamental to many ecological studies in forestry. Such assessment is, in fact, crucial for uncovering how trees interact, compete for resources, and respond to environmental factors, thereby shaping forest structure and resilience  \citep{pommerening2013mark, pommerening2019individual}. Among these attributes, diameter at breast height (dbh) and the trees' height are key indicators of the trees' growth, health, and overall forest dynamics.
For instance, nearby trees may exhibit similar sizes due to localised competition for resources, indicating association, whereas the variation in tree growth rates across a forest may reflect environmental gradients such as moisture availability or sunlight exposure, revealing variation that helps identify habitat quality and forest health over space \citep{stoyan2000recent}. However, understanding these attributes and their space-dependent associations/variations is inherently linked to the spatial distribution of the trees, which is often uneven due to underlying environmental factors.

In this paper, our primary motivation is to investigate the potential space-dependent associations/variations in dbh and tree height across two different ecosystems:  Longleaf pine (Pinus palustris) trees in southern Georgia, USA, and Scots pine (Pinus sylvestris L) trees in Pfynwald, Switzerland. In the former case, Longleaf pines are marked by their dbh \citep{platt1988population}, while in the latter, Scots pines are marked by their height \citep{pfynwald:2016}.
Our aim is to offer deep ecological insights that go beyond simply studying the trees' locations. 
Tree locations and their associated attributes are inherently realisations of marked spatial point processes, where spatial points refer to the location of trees and marks refer to the attributes attached to trees.
In practice, such marked point processes are increasingly becoming heterogeneous, requiring methodologies that account for the spatial inhomogeneity in the locations of trees while studying the association/variation among their attributes. 
This is, in particular, the case of both the Longleaf and the Scots pine trees, as they exhibit significant inhomogeneity, reflecting inhomogeneous spatial distributions across the forest.
However, for such real-valued attributes, despite the need for suitable methods controlling for the inhomogeneity of the points, 
existing approaches to study mark association/variation, including various mark correlation functions \citep{
cressie93, StoyanStoyan1994, Illian2008, Baddeley2010}, are theoretically tied with the assumption of stationarity, an assumption that rarely holds in real applications \citep{Eckardt:Moradi:currrent}.
In essence, any such mark correlation function aims at measuring the association/variation among tree attributes, under the assumption that, for a fixed interpoint distance, each tree has an equal number of neighbouring trees.
Given this strong underlying assumption, which rarely holds in practice, the use of these mark correlation functions in settings with evident inhomogeneity in the distribution of the trees, such as the case of both Longleaf and Scots pines, might therefore lead to misleading/incomplete/poor conclusions about the mark association/variation among trees. 
Aiming to address this clear limitation, we first develop novel mark correlation functions that explicitly account for the complex inhomogeneity in the spatial distribution of trees, tending to be valuable tools for general biological and environmental applications.

Reviewing the literature on point processes with real-valued marks indicates that mark correlation functions generally fall into two main categories based on their underlying objective: those aimed at detecting significant mark association, capturing dependencies between marks as a function of pairwise spatial distances, and those focused on identifying significant mark variation, revealing systematic spatial patterns or trends in the marks. 
The former includes Isham's  \citep{isham1985marked},  
Stoyan's \citep{StoyanStoyan1994}, Beisbart's and Kerscher's mark correlation functions \citep{Beisbart:2000}, and Shimatani's and Schlather's $I$ functions \citep{Shimatani:MoranI, Schlather2004}. The latter includes mark covariance function  \citep{DBLP:journals/eik/Stoyan84}, mark variogram  \citep{cressie93
}, and mark differentiation function   \citep{20113358596}, 
which, even though they have distinct viewpoints, describe the average mark variation with respect to pairwise spatial distances between points. 
These functions were initially developed for stationary marked point processes in $\R^2$, mostly applied in environmental contexts such as forest ecology and forest fires. However, \cite{Eckardt:Moradi:currrent} proposed their counterpart for stationary marked point processes on linear networks, with applications to street trees. 
In real-world biological and environmental systems, the stationary assumption often fails, as the intensity of events frequently varies due to habitat heterogeneity, environmental gradients, or external disturbances.
Consequently, spatial variations in the distribution of the points may induce spurious indications of mark association, thereby obscuring the true underlying relationships and potentially resulting in erroneous inferences about mark association/variation. 

The problem with mark correlation functions for stationary marked point processes mirrors the difficulties encountered during the early developments of general summary statistics for unmarked point processes, such as Ripley's $K$-function, pair correlation function, and $J$-functions, which originally assumed spatial homogeneity before further advancements to address spatial heterogeneity. We also note that a mark-weighted version of Ripley's $K$-function was introduced by \cite{penttinen1992marked}, whose interpretation turns out not to be as easy as it merged two aspects: the spatial interaction between points in terms of clustering/repulsion, and association/variation among marks. 
\cite{InhomK2000} introduced a new class of point processes, called second-order intensity-reweighted stationary, for which inhomogeneous versions of the $K$-functions and pair correlation functions were defined. Similarly, \cite{van11} introduced the much-limited class of intensity-reweighted moment stationary point processes and developed an inhomogeneous version of the $J$-functions tailored to this framework. Poisson point processes naturally belong to these classes, and Cox point processes also fall within these categories, subject to specific conditions on the covariance function being satisfied.
Moreover, \cite{Ang:Baddeley:Nair2012} presented inhomogeneous $K$- and pair correlation functions for unmarked point processes on linear networks, and showed real applications where the inhomogeneous ones are more precise as they account for the underlying spatial distribution of points.

A detailed look at the construction of mark correlation functions for stationary marked point processes reveals that they are, in fact, formulated as the ratio of two pair correlation functions: the mark-weighted homogeneous pair correlation function in the numerator and the standard homogeneous pair correlation function in the denominator. This observation inspires our proposal to weight the existing mark correlation functions by a reciprocal factor based on the intensity function, generalising mark correlation functions to non-stationary settings. 
We propose non-parametric estimators, and design simulation studies with various scenarios in which the spatial distribution of points ranges from inhomogeneous patterns with no spatial interaction to inhomogeneous patterns exhibiting clustering tendencies. Within these scenarios, we consider distinct mark functions which impose mark association/variation among points. Our simulation studies, designed to evaluate the robustness and effectiveness of the proposed estimators with respect to the power of the test and type I error probability, demonstrate their applicability to a wide range of settings where spatial inhomogeneity is evident. Notably, our proposed inhomogeneous mark correlation functions outperform the existing ones not only in terms of the power of the test and type I error rates, but also in capturing the positivity/negativity and spatial range of mark association/variation.

The paper is organised as follows. 
Section \ref{sec:pre} provides some general details on marked point processes, followed by mark correlation functions for stationary point processes in general state spaces having real-valued marks.
In Section \ref{sec:inhommarkcorr}, we introduce inhomogeneous mark correlation functions for the same contexts and present their non-parametric estimators.
Section \ref{sec:simu} presents a simulation study in which we evaluate the performance of our proposed inhomogeneous mark correlation functions in various scenarios where the spatial distribution of points is not uniform, and mark association/variation is evident.
Section \ref{sec:apps} showcases our findings for Longleaf pine trees in southern Georgia, USA, and Scots pine (Pinus sylvestris L) forest trees in Pfynwald, Switzerland.
The paper ends with a discussion in Section \ref{sec:diss}.

\section{Preliminaries}\label{sec:pre}

Let $X=\{(x_i, m(x_i))\}_{i=1}^{N}, N<\infty,$ be a marked point process in a general state space $\mathbb{S}$, which is equipped with a Lebesgue measure defined by $|A| = \int_{A} \de u$, for $A \subseteq \mathcal{B}(\mathbb{S})$, where $\mathcal{B}$ is the collection of Borel sets, and endowed with a spatial distance metric $d_{\mathbb{S}}$. 
The general state space $\mathbb{S}$ is often the product of two spaces, one for the spatial locations $\{x_i\}_{i=1}^{N}$ denoted by $\mathbb{S}^{\mathrm{space}}$ and one for the corresponding marks $\{m(x_i)\}_{i=1}^{N}$ denoted by $\mathbb{S}^{\mathrm{mark}}=\mathbb{M}$ which is a Polish space endowed with a reference measure $\nu$ on its Borel $\sigma$-algebra $\mathcal{B}(\mathbb{M})$.
If the spatial locations live in $\mathbb{S}^{\mathrm{space}} = \R^2$ and marks are positive real numbers, we write $\mathbb{S} = \R^2 \times \R^+$ and choose $d_{\mathbb{S}}$ to be Euclidean distances, whereas if the spatial locations are on a linear network $\mathbb{L} \subset \R^2$, then $d_{\mathbb{S}}$ is often the shortest-path distance.
If the distribution of $X$ is translation invariant, meaning that, for every $a \in \mathbb{S}^{\mathrm{space}}$, $X$ and $X + a = \{(x_i +a, m(x_i))\}_{i=1}^{N}$ have the same distribution, then $X$ is called to be stationary. While this definition holds for $\mathbb{S}^{\mathrm{space}} = \mathbb{R}^d$, $d \geq 1$, a different approach is required when $\mathbb{S}^{\mathrm{space}} = \mathbb{L}$, as standard transformations may not preserve the structure of the linear network $\mathbb{L}$; see \citet{cronie2020inhomogeneous, Eckardt:Moradi:currrent} for details.

According to Campbell's formula, for any non-negative measurable function $h: \mathbb{S} \to  \R$, 
\begin{align*}
& 
\e 
\left[
\sum_{(x,m(x)) \in X}
h
(x,m(x))
\right]
= 
\int_{\mathbb{S}^{\mathrm{mark}}}
\int_{\mathbb{S}^{\mathrm{space}}}
h
(u,m(u))
\lambda
(u, m(u))
\de u \nu(\de m(u)),
\end{align*}
where $\lambda(\cdot,\cdot)$ is the intensity function of $X$, governing its distribution \citep{Daley2003}.
In fact, the intensity function $\lambda (u, m(u))$ gives the expected number of points per unit-size area 
in the vicinity of $(u, m(u)) \in \mathbb{S}$. If the intensity function is not constant, the point process $X$ is said to be inhomogeneous. The above Campbell’s formula can be extended to higher orders, yielding expressions that involve the $n$th-order product density functions $\lambda^{(n)}$, which describe the joint occurrences of $n$-tuples of points. When $n=2$, the second-order product density functions $\lambda^{(2)}$ takes the form
\begin{align*}
\lambda^{(2)}
   \big(
    (u, m(u)), (v, m(v))
   \big)  
   =
   f^{(2)}
   (
   m(u), m(v)
   | u, v
   )
   \lambda^{(2)} (u,v),
\end{align*}
where $\lambda^{(2)} (u,v)$ is the second-order product density function of the ground process, and $f^{(2)}$ is the conditional density of marks. From here on, when such functions are written with arguments consisting solely of spatial locations or solely of marks, they will be understood to refer to the ground process or the mark process, respectively.
A less restrictive condition than stationarity refers to second-order intensity-reweighted stationary, which means the pair correlation function
\begin{align*}
    g
    \big(
    (u, m(u)), (v, m(v))
    \big)
    &=
    \frac{
    \lambda^{(2)}
    \big(
    (u, m(u)), (v, m(v))
    \big)
    }{
    \lambda
    (u, m(u))
    \lambda
    (v, m(v))
    }
    =
    \frac{
    f^{(2)}
   \big(
   m(u), m(v)
   | u, v
   \big)
    }{
    f
   \big(
   m(u)
   | u
   \big)
   f
   \big(
   m(v)
   | v
   \big)
    }
    \frac{
    \lambda^{(2)} (u,v)
    }{
    \lambda (u)
    \lambda (v)
    }
    \\
    &=
    \gamma^{(2)}
    \big(
    m(u), m(v) \big| u, v
    \big)
    g(u,v),
\end{align*}
only depends on the spatial distance between the points. In other words,
\begin{align}\label{eq:2ndorder}
    g
    \big(
    (u, m(u)), (v, m(v))
    \big)
    =  
    \bar{g} 
    (
    d_{\mathbb{S}} (u, v)
    )
    \gamma^{(2)}
    \big(
    m(u), m(v) \big| u, v
    \big),
\end{align}
where $\bar{g}: \R \to \R$ \citep{InhomK2000, cronie2024discussion}; note the distinctions between the cases $\mathbb{S}^{\mathrm{space}} = \R^d$ and $\mathbb{S}^{\mathrm{space}} = \mathbb{L}$ \citep{cronie2020inhomogeneous, Eckardt:Moradi:currrent}. 
Summary statistics such as the inhomogeneous $K$- and pair correlation functions are defined for this class of point processes. Note the relationship between the $K$- and pair correlation functions according to which, $g(r) = K^{\prime}(r)/(2\pi r)$, if $\mathbb{S}^{\mathrm{space}} = \R^2$, where $K^{\prime}(r)$ is the derivative of the $K$-function. 

\subsection{Mark correlation functions for stationary point processes}

{
Prior to introducing mark correlation functions for inhomogeneous marked point processes, which are given in Section \ref{sec:inhommarkcorr}, it is essential to review the underlying construction principle of these functions for the stationary marked point processes.
} For every two points $(x, m(x)), (y, m(y)) \in X$, conditional on having an interpoint distance $d_{\mathbb{S}}(x, y)=r$, unnormalised mark correlation functions, formalised as a conditional expectation in the Palm sense, are defined as
\begin{align}\label{eq:ctf}
    c_{\tf_f}(r)
    =
    \e 
    \left[
    \tf_f 
    \left(
    m(x), m(y)
    \right)
    \Bigl\vert 
    (x, m(x)), (y, m(y)) \in X
    \right],
\end{align}
where $\tf_f: \mathbb{M} \times \mathbb{M} \to \R^+$ is a so-called test function in its most general form. 
Given the conditional nature of $c_{\tf_f}(r)$, {\eqref{eq:ctf} is commonly rewritten in the form of}
\begin{align}\label{eq:ctf:2ndorderproducts:ratio}
c_{\tf_f}(r)
=
\frac{
\lambda^{(2)}_{\tf_f}
(r)}{
\lambda^{(2)}
(r)},
\end{align}
where $\lambda^{(2)}_{\tf_f}$ is the density of the factorial moment measure
\begin{align*}
\alpha_{\tf_f}(B_1, B_2)
=
\e
\bigg[
\sum_{(x,m(x)), (y, m(y)) \in X}^{\neq}
\tf_f 
(
m(x), m(y)
)
\1_{B_1} \{ x \}
\1_{B_2} \{ y \}
\bigg], 
\quad
B_1, B_2 \in \mathcal{B}(\mathbb{S}^{\mathrm{space}}),
\end{align*}
and $\lambda^{(2)}
(r)$ is the second-order product density function of the ground process, i.e., the unmarked version of $X$ \citep[Chapter 5]{Illian2008}. {Importantly, recalling the definition of the pair correlation and noting that constant terms cancel out, we point to the fact that \eqref{eq:ctf:2ndorderproducts:ratio} can also be formalised as the ratio of two pair correlation functions.} In other terms, regardless of the choice $\tf_f$, the unnormalised mark correlation function $c_{\tf_f}(r)$ is just the ratio of two homogeneous pair correlation functions: the  $\tf_f$-weighted homogeneous pair correlation function of the marked point process $X$
and the homogeneous 
pair correlation function of the ground process, i.e. the unmarked version of $X$. Note that if one writes  \eqref{eq:ctf:2ndorderproducts:ratio} as such ratio, the denominators cancel each other out, and thus we get the ratio based on product density functions.
Under 
mark independence, {which is assumed} 
when $r\rightarrow\infty$, we have
\begin{eqnarray}\label{eq:normfact}
    c_{\tf_f}
    =
    c_{\tf_f}(\infty)
    =
    \int_{\mathbb{M}}
    \int_{\mathbb{M}}
    \tf_f(m(x), m(y)) \nu(\de m (x)) \nu(\de m(y)).
\end{eqnarray}
which consequently gives rise to the normalised mark correlation functions
\begin{align}
\kappa_{\tf_f}(r)
=
\frac{
c_{\tf_f}(r)
}{
c_{\tf_f}
},
\end{align}
with a value of one under the mark independence assumption. Thus, any significant deviation of $\kappa_{\tf_f}(r)$ from one indicates positive/negative mark association/variation depending on the test function $\tf_f$ employed. 
The most frequently used test functions $\tf_f$ include 
\begin{align}
    \tf_f 
    \left(
    m(x), m(y)
    \right)
    =
    m(x) m(y),
\end{align}
and
\begin{align}
 \tf_f 
    \left(
    m(x), m(y)
    \right)
    =
    0.5 (m(x) - m(y))^2,   
\end{align}
establishing Stoyan's mark correlation function $\kappa_{mm}$ and the mark variogram $\gamma_{mm}$ with normalisation factors 
$\mu_m^2$, where $\mu_m =\sum_{(x,m(x)) \in X} m(x)/N$, and $\sigma^2_m = \sum_{(x,m(x)) \in X} (m(x) - \mu_m)^2/(N-1)$ \citep{cressie93, StoyanStoyan1994, LIMA}.  
A review and a simulation-based comparison of all well-known such test functions can be found in 
\cite{Eckardt2024Rejoinder}.

While in the literature, following different perspectives, there are various mark correlation functions in the form of \eqref{eq:ctf}, all such functions require the underlying processes to be stationary with test functions $\tf_f$ just being functions of marks, ignoring the potential heterogeneity in the spatial distribution of points. 
Consequently, since real-world datasets are rarely uniformly distributed, conclusions about mark association and/or variation derived from the mark correlation functions $c_{\tf_f}(r)$ and $\kappa_{\tf_f}(r)$ may lack completeness.

\section{Mark correlation functions for inhomogeneous point processes}\label{sec:inhommarkcorr}

In this section, recalling \eqref{eq:2ndorder}, we let $X=\{(x_i, m(x_i))\}_{i=1}^{N}, N<\infty$, be a second-order intensity-reweighted stationary marked point process. 
To account for the spatial distribution of points, we propose to employ a weight factor $w(x,y)=1/\lambda(x)\lambda(y)$, in which $\lambda(x)$, here, refers to the intensity function of the unmarked version of $X$ evaluated at $x \in \mathbb{S}^{\mathrm{space}}$, compensating for the inhomogeneity in the spatial distribution of points. 
More specifically,  we redefine \eqref{eq:ctf} as 
\begin{align}\label{eq:ctfinhom}
    c_{\tf_f}^{\textrm{inhom}}(r)
    =
    \e_w
    \left[
     \tf_f 
    \left(
    m(x), m(y)
    \right)
    \bigg | 
    (x, m(x)), (y, m(y)) \in X
    \right],
\end{align}
where $\e_w$ is a conditional expectation with respect to a weighted Palm distribution. Therefore, \eqref{eq:ctfinhom} can be rewritten as
\begin{align}
c_{\tf_f}^{\textrm{inhom}}(r)
=
\frac{
\lambda^{(2)}_{w, \tf_f}
(r)}{
\lambda^{(2)}_{w}
(r)},
\end{align}
where $\lambda^{(2)}_{w, \tf_f}$ is the density of the $w$-weighted factorial moment measure $\alpha^{\textrm{inhom}}_{\tf_f}(\cdot, \cdot)$ given as
\begin{align*}
\alpha^{\textrm{inhom}}_{\tf_f}(B_1, B_2)
=
\e
\bigg[
\sum_{(x,m(x)), (y, m(y)) \in X}^{\neq}
\frac{
\tf_f 
(
m(x), m(y)
)
\1_{B_1} \{ x \}
\1_{B_2} \{ y \}
}{
\lambda(x) \lambda(y)
}
\bigg], 
\quad
B_1, B_2 \in \mathcal{B}(\mathbb{S}^{\mathrm{space}}),
\end{align*}
and $\lambda^{(2)}_w$ is the ($w$-weighted) second-order product density function of the ground process.
Interestingly, this shows that unnormalised inhomogeneous mark correlation functions can be naturally expressed as the ratio of the mark-weighted inhomogeneous pair correlation function to the inhomogeneous pair correlation function of the ground process, closely aligning with their definition in the stationary case. 
Further, given 
the theoretical relationships between the pair correlation,  
$K$- and $J$-functions \citep{InhomK2000, van11}, this 
formulation extends to inhomogeneous $K$- and $J$-functions in the marked settings. Therefore, one can alternatively define novel mark correlation functions based on inhomogeneous $K$- and $ J$-functions, opening novel thoughts to interpret results based on their mark-weighted versions. 
For instance, in the case of the inhomogeneous $K$-function, one can define an unnormalised inhomogeneous mark correlation function in the form of
\begin{align}\label{eq:ctfK}
    \frac{
    K^{\mathrm{inhom}}_{\tf_f}(r)
    }{
    K^{\mathrm{inhom}}(r)
    }
    =
    \frac{
    \e 
    \big[
    \sum_{(x,m(x)) \in X} 
    \tf_f 
    (
    m(x), m(y)
    )
    \1
    \{
    d_{\mathbb{S}}(x,y) \leq r
    \}
    e(x,y)
    \big/
    \lambda(x)
    \big|
    (y, m(y)) \in X
    \big]
    }{
    \e
    \big[
    \sum_{(x,m(x)) \in X}
    \1
    \{
    d_{\mathbb{S}}(x,y) \leq r
    \}
    e(x,y)
    \big/
    \lambda(x)
    \big|
    (y, m(y)) \in X
    \big]
    },
\end{align}
where $e(x,y)$ is an edge correction compensating for the lack of information near the boundary of the observation window \citep{InhomK2000}, and $\1$ is an indicator function. 

\begin{lemma}\label{lem:lemma}
Let $X$ be a second-order intensity-reweighted stationary marked point process in $\mathbb{S}$. Given that, in \eqref{eq:ctfinhom}, $d_{\mathbb{S}}(x,y) = r$ and the test function is of the form $\tf_f (m(x),m(y)) = m(x) m(y)$, the normalising factor becomes $\mu_m^2$, where $\mu_m =\sum_{(x,m(x)) \in X} m(x)/N$.
\end{lemma}
\begin{proof}
By employing Campbell's formula, we can write 
\begin{align*}
\e
[
\lambda^{(2)}_{w, \tf_f}
(r)
]
&=
\int \int \int \int
\frac{
m(u)
m(v)
\1 \{ d_{\mathbb{S}}(u,v) = r \}
\lambda^{(2)}
(
(u, m(u)),(v, m(v))
)
}{
\lambda(u)
\lambda(v)
}
\de u
\de m (u)
\de v
\de m(y)
\\
&=
\int \int \int \int
\frac{
m(u)
m(v)
\1 \{ d_{\mathbb{S}}(u,v) = r \}
f^{(2)}
(
m(u), m(v) | u,v
)
\lambda^{(2)}
(
(u),(v)
)
}{
\lambda(u)
\lambda(v)
}
\de u
\de m (u)
\de v
\de m(y)
\\
&=
\int \int \int \int
m(u)
m(v)
\1 \{ d_{\mathbb{S}}(u,v) = r \}
f^{(2)}
(
m(u), m(v) | u,v
)
g(u,v)
\de u
\de m (u)
\de v
\de m(y)
\\
&=
\int \int \int \int
m(u)
m(v)
\1 \{ d_{\mathbb{S}}(u,v) = r \}
f(m(u))
f(m(v))
g(u,v)
\de u
\de m (u)
\de v
\de m(y)
\\
&=
\int m(u) f(m(u)) \de m (u)
\int m(v) f(m(v)) \de m (v)
\int \int
\1 \{ d_{\mathbb{S}}(u,v) = r \} g(u,v) \de u \de v
\\
&=
\mu^2_m
\int \int
\1 \{ d_{\mathbb{S}}(u,v) = r \}
g(u,v) \de u \de v,
\end{align*}
where, under mark independence, we can write
$
f^{(2)}
(
m(u), m(v) | u,v
)
=
f(m(u))
f(m(v))
$. Similarly, we can show that 
\begin{align*}
 \e
[
\lambda^{(2)}_{w}
(r)
]
=
\int \int
\1 \{ d_{\mathbb{S}}(u,v) = r \}
g(u,v) \de u \de v,
\end{align*}
leading to the final result that when the test function is $\tf_f (m(x),m(y)) = m(x) m(y)$, the normalising factor will become $\mu_m^2$. Note that, in practice, similar to the case of pair correlation functions,  $\1 \{ d_{\mathbb{S}}(u,v) = r \}$ will be replaced by $\1 \{ r - \epsilon \leq d_{\mathbb{S}}(u,v) \leq r + \epsilon \}$ for some very small $\epsilon >0$.
\end{proof}

In a similar manner as in Lemma \ref{lem:lemma}, one can show that if in \eqref{eq:ctfinhom}, the test function takes the form $\tf_f (m(x),m(y)) = 0.5(m(x)-m(y))^2$, establishing the inhomogeneous mark variogram, the normalisation factor becomes the variance of the marks. We emphasise that the normalisation factors will be the same if one instead of \eqref{eq:ctfinhom} uses the $K$-function-based mark correlation function given in \eqref{eq:ctfK}.
Note that, under homogeneity, where the intensity function is constant, both \eqref{eq:ctfinhom} and \eqref{eq:ctfK} reduce to their classical version originally defined for homogeneous marked point processes.

\subsection{Estimation}
This section is devoted to presenting estimators for $c_{\tf_f}^{\textrm{inhom}}(r)$ when the marked point process $X$ is observed in a window $W \subset \mathbb{S} = \R^2 \times \R^+$; extending the estimators to other state spaces will be straightforward. 
Unbiased estimators of terms involved in $c_{\tf_f}^{\textrm{inhom}}(r)$ are given as
\begin{eqnarray}\label{eq:est1}
\widehat{\lambda}^{(2)}_{w, \tf_{f}}(r)
&=&
\frac{1}{2\pi r |W|}
\sum_{(x, m(x)), (y, m(y)) \in X}^{\neq}
\frac{
\tf_{f}
\big(
m(x),m(y)
\big)
\mathcal{K}
\big(
d_{\mathbb{S}}(x,y) - r
\big)
e (x,y)
}{
\lambda (x)  \lambda (y)
}, 
\end{eqnarray}
and 
\begin{eqnarray}\label{eq:est2}
\widehat{\lambda}^{(2)}_{w}(r)
&=&
\frac{1}{2\pi r |W|}
\sum_{(x, m(x)), (y, m(y)) \in X}^{\neq}
\frac{
\mathcal{K}
\big(
d_{\mathbb{S}}(x,y) - r
\big)
e (x,y)
}{
\lambda (x)  \lambda (y)
}, 
\end{eqnarray}
leading to the ratio-unbiased estimator
\begin{align}\label{eq:est3}
\widehat{c}_{\tf_f}^{\textrm{inhom}}(r) 
=
\frac{
\widehat{\lambda}^{(2)}_{w, \tf_{f}}(r)
}{
\widehat{\lambda}^{(2)}_w(r)
},
\end{align}
where $\mathcal{K}(\cdot)$ is a kernel function, $e(x, y)$ is an edge correction factor,
and $|W|$ is the area of the observation window $W$. Taking a closer look
at the estimator \eqref{eq:est2}, one can see that this is in fact the unbiased estimator of the inhomogeneous pair correlation function of the ground process \citep{InhomK2000},  with \eqref{eq:est1} being its mark-weighted version. Note that, in \eqref{eq:est3}, the factor $1/ 2\pi r |W|$ vanishes. Looking at the estimators for general summary statistics for point processes, the frequently considered edge correction factors are known as `translation' and `Ripley' \citep[Chapter 7]{Baddeley2015}. We add that the form of edge corrections varies depending on the state space $\mathbb{S}^{\mathrm{space}}$; see \cite{moradi2018spatial} for further details on this matter when  $\mathbb{S}^{\mathrm{space}} =  \mathbb{L}$.

In practice, the intensity function in \eqref{eq:est1} and \eqref{eq:est2} is unknown, and thus, we need to estimate it from the observed point pattern.
Within the literature, various kernel- and Voronoi-based intensity estimators are presented \citep{Baddeley2015, MoradiVor2019}. 
The two most frequently used kernel-based intensity estimators for unmarked spatial point patterns are
\begin{equation}
\label{e:kde.2D.unif}
\widehat \lambda^{\text{U}}(u)
= 
\frac{1}{c_{W}(u)}
\sum_{(x, m(x)) \in X}
\mathcal{K}(u - x),
\quad u \in W \subset \mathbb{S}^{\mathrm{space}},
\end{equation}
and
\begin{equation}
\label{e:kde.2D.JD}
\widehat \lambda^{\text{JD}}(u)
= 
\sum_{(x, m(x)) \in X}
\frac{\mathcal{K}(u - x)}{c_{W}(x)},
\quad u \in W \subset \mathbb{S}^{\mathrm{space}},
\end{equation}
where \(\mathcal{K}\) is a kernel function, and
\begin{equation}
\label{e:cW}
c_{\sigma,W}(u)
= 
\int_W \kappa_{\sigma}(u - v) \mathrm{d}v,
\quad u \in W \subset \mathbb{S}^{\mathrm{space}},
\end{equation}
is an edge correction factor compensating for the unseen data outside the observation window. Several bandwidth selection approaches, from different perspectives, are available within the literature to choose the bandwidth parameter associated with the kernel function $\mathcal{K}$. Note that, regardless of the choice of bandwidth, the estimator \eqref{e:kde.2D.unif} is unbiased under homogeneity, while the estimator \eqref{e:kde.2D.JD} conserves the mass. In the case of Voronoi-based intensity estimators, we have the resample-smoothed estimator
\begin{align}
\label{SmoothVor}
\widehat{\lambda}_{p,m}^{V}(u)
=
\frac{1}{m}\sum_{i=1}^m
\frac{\widehat{\lambda}_i^{V}(u)}{p}
,\quad
\ u\in W \subset \mathbb{S}^{\mathrm{space}},
\end{align}
which is based on averaging $m$ rescaled Voronoi estimators of $p$-thinned point processes obtained from the point process $X$. Here, \(\widehat{\lambda}_i^{V}(u)\) corresponds to the Voronoi-based intensity estimator of the \(i\)-th $p$-thinned pattern \citep{MoradiVor2019}. This estimator is unbiased under the assumption of homogeneity and satisfies the mass conservation property simultaneously. A review of different kernel- and Voronoi-based intensity estimators, together with discussing various bandwidth selection methods, can be found in \cite{mateu2024non}.

\section{Simulation study}\label{sec:simu}

In this section, we evaluate the ability of our proposed inhomogeneous mark correlation functions to detect mark associations/variations across a range of scenarios, specifically assessing deviations from the assumption of random labelling. We consider inhomogeneous marked point processes observed within the unit square window $[0,1]^2$, focusing on scenarios where there is either i) mark association, or ii) mark variation among points. For scenarios involving mark association, we use the test function $\tf_f(m_1,m_2) = m_1 m_2$ with $m_1=m(x)$ and $m_2=m(y)$, giving rise to the inhomogeneous mark correlation function $\kappa^{\mathrm{inhom}}_{mm}$. In contrast, to capture mark variation, we apply the test function $\tf_f(m_1,m_2) = 0.5 (m_1 - m_2)^2$, giving rise to the inhomogeneous mark variogram $\gamma^{\mathrm{inhom}}_{mm}$. For each scenario, we generate 100 point patterns, and, using global envelop tests with 1000 permutations \citep{myllymaki2017global}, we study the power of the test, type I error probability, as well as the range and positivity/negativity of detected mark associations/variations. Concerning the study of type I error probability, we use the same generated point patterns in each scenario, but the marks will be replaced by random numbers generated between 0 and 1, representing random labelling. Moreover, we compare the performance of our proposed inhomogeneous mark correlation functions with that of the homogeneous ones, denoted by $\kappa_{mm}$ and $\gamma_{mm}$. We add that global envelop tests are performed using the completely non-parametric rank envelope test based on extreme rank lengths, known as `\textsf{erl}', at the significance level $0.05$. As with practical applications, the intensity function must be estimated. Therefore, in our simulation studies, we also estimate the intensity function. More specifically, we use the kernel-based intensity estimator \eqref{e:kde.2D.JD} combined with Cronie and van Lieshout's criterion for choosing the smoothing bandwidth parameter \citep{cronie2018non}. 

\subsection{Mark association}\label{sec:asso}

Here, we consider the mark distribution $m(u)=m(x,y)=\sin(x^2+y^2)$ which induces associations among the marks of nearby points. Under this framework, we explore two scenarios in which the spatial distribution of points ranges from inhomogeneous Poisson processes to Log-Gaussian Cox processes. Given the considered mark distribution and the spatial distribution of points, on average, a positive association for moderate ranges of spatial distances is expected.

\subsubsection{Inhomogeneous Poisson point processes}

We consider an inhomogeneous Poisson point process,
with an intensity function given by $\lambda(u)=\lambda(x,y)= 50\exp(\sin(4x^2 + 4y^2))$. Regarding the power of the test, we find that the inhomogeneous mark correlation function $\kappa_{mm}^{\textrm{inhom}}$ detects the association of the mark for all 100 simulated point patterns, while the homogeneous mark correlation function $\kappa_{mm}$ identifies it only in $15\%$ of the cases. As a showcase, in Figure \ref{fig:assoinhomPoiss}, we show the results of global envelope tests for one of the 100 simulated patterns. From the pattern, one can see that nearby points have similar marks, highlighting their association, whereas the marks become more dissimilar if we increase the spatial interpoint distance. However, $\kappa_{mm}$ shown in the right plot does not identify any deviation from random labelling. In contrast, $\kappa_{mm}^{\textrm{inhom}}$ clearly identifies the positive mark association among points, representing the true mark association among points. In this particular point pattern, due to the inhomogeneity of the pattern, there is a higher frequency of pairs of points with small mark products at very short distances, compared to those with larger products. Thus, $\kappa_{mm}^{\textrm{inhom}}$ stays within the envelope for very small spatial distances, reflecting the empirical spatial/mark distribution of points. Note also the differences in their fluctuations across different spatial distance ranges.
\begin{figure}[!h]
    \centering
    \includegraphics[scale=0.19]{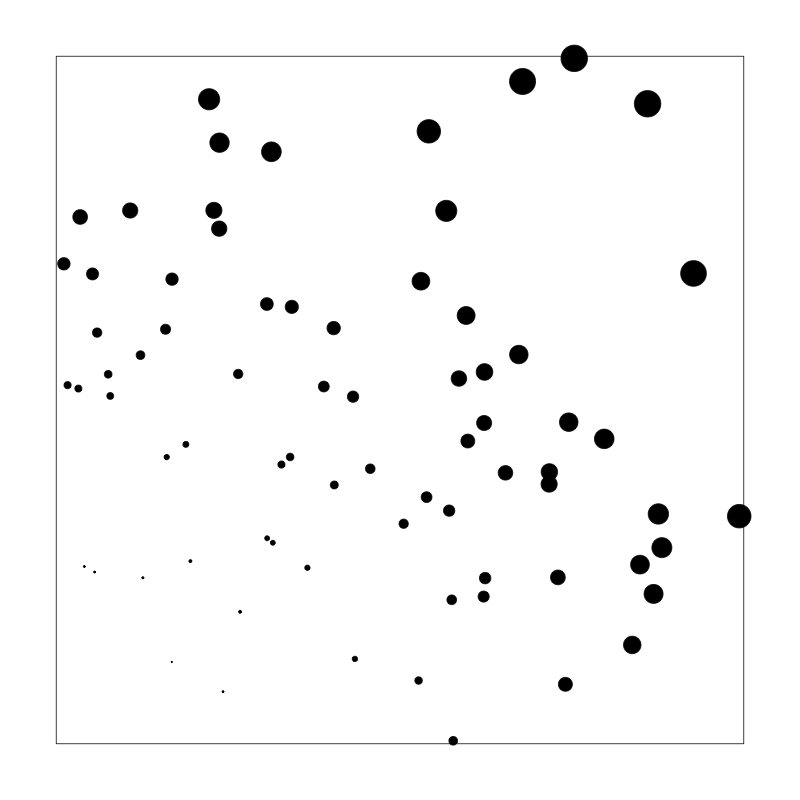}
    \includegraphics[scale=0.18]{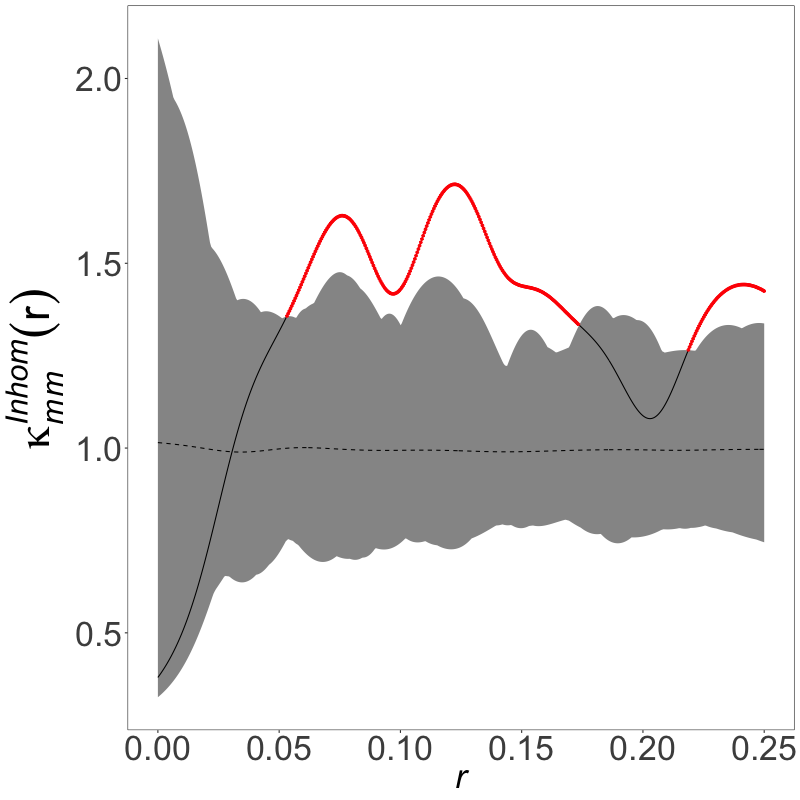}
    \includegraphics[scale=0.18]{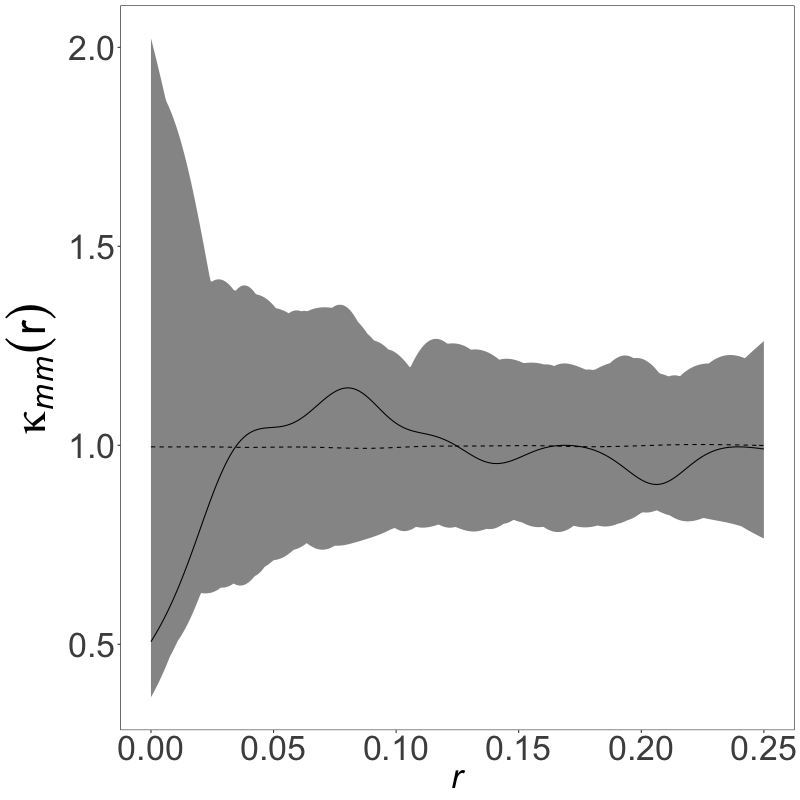}
    \caption{
    One of the 100 simulated point patterns, together with its corresponding global envelope tests based on $\kappa_{mm}^{\textrm{inhom}}$ and $\kappa_{mm}$.
    }
    \label{fig:assoinhomPoiss}
\end{figure}

Turning to the type I error probability, we find that both $\kappa_{mm}^{\textrm{inhom}}$ and $\kappa_{mm}$ mistakenly detected some mark association in $9\%$ of the patterns. Going through the patterns identified as having significant mark associations by $\kappa_{mm}^{\textrm{inhom}}$ and $\kappa_{mm}$, we observe that for some of them, estimated mark correlation functions are only slightly surpassing the envelope for very short ranges of distances. Otherwise, their type I error probability stays around $5\%$.

\subsubsection{Log-Gaussian Cox point processes}

Here, we consider a log-Gaussian Cox point process where the driving Gaussian random field has a mean function $(x,y) \mapsto \log(90) + \sin(4x^2 + 4y^2) - 1$, and an exponential covariance function $((x_1,y_1),(x_2,y_2)) \mapsto 1.5 \exp ( - || (x_1,y_1) - (x_2,y_2) || / 0.12 )$. Hereby, the intensity function is given as $\lambda(u)=\lambda(x,y)=90 \exp ( \sin(4x^2 + 4y^2) - 0.25)$. In this scenario, we find that $\kappa_{mm}^{\textrm{inhom}}$ correctly identifies the assigned mark association in $92\%$ of the patterns, whereas $\kappa_{mm}$ is successful only in $53\%$ of the cases. Figure \ref{fig:assoinhomlgcp} shows the results for one of the simulated point patterns. From the point pattern, the inhomogeneity and mark association can be observed. For this specific pattern, $\kappa_{mm}$ completely stays inside the envelope and detects no deviation from random labelling for any range of distances. In contrast, $\kappa_{mm}^{\textrm{inhom}}$ correctly detects the positive mark association assigned to the points for different ranges of spatial distances, mimicking the spatial/mark distribution of the points. The type one error probabilities for $\kappa_{mm}^{\textrm{inhom}}$ and $\kappa_{mm}$ are estimated as $6\%$ and $4\%$, respectively. Also, it seems that $\kappa_{mm}^{\textrm{inhom}}$ have a lower variation than $\kappa_{mm}$.

\begin{figure}[!h]
    \centering
    \includegraphics[scale=0.19]{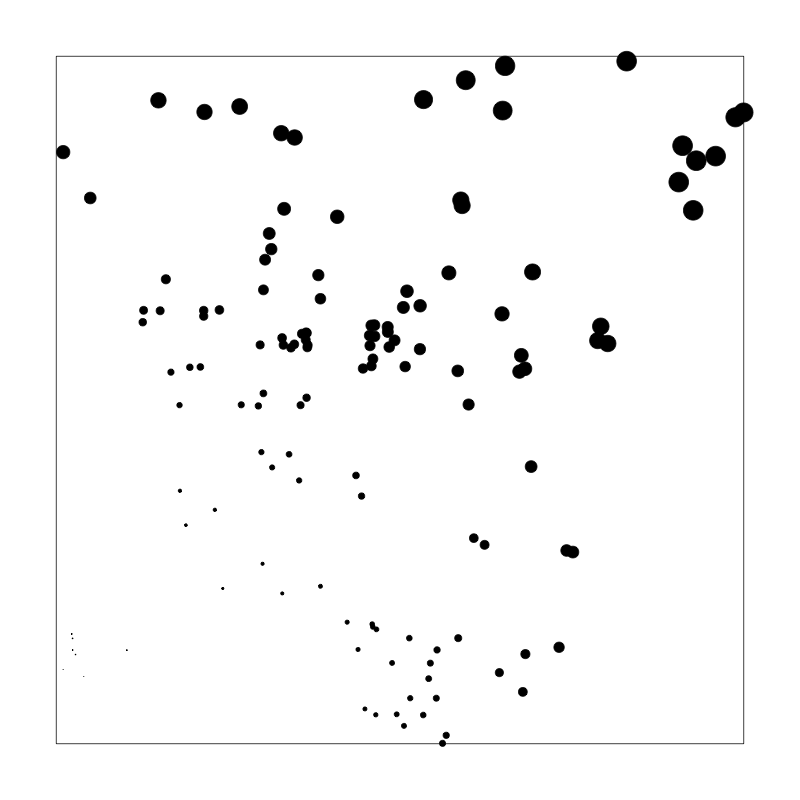}
    \includegraphics[scale=0.18]{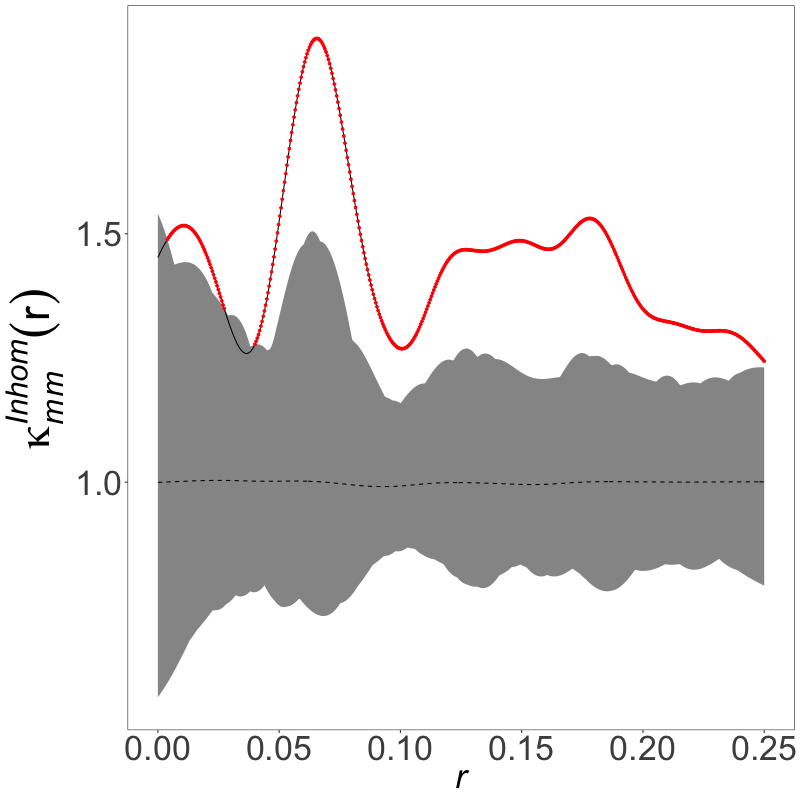}
    \includegraphics[scale=0.18]{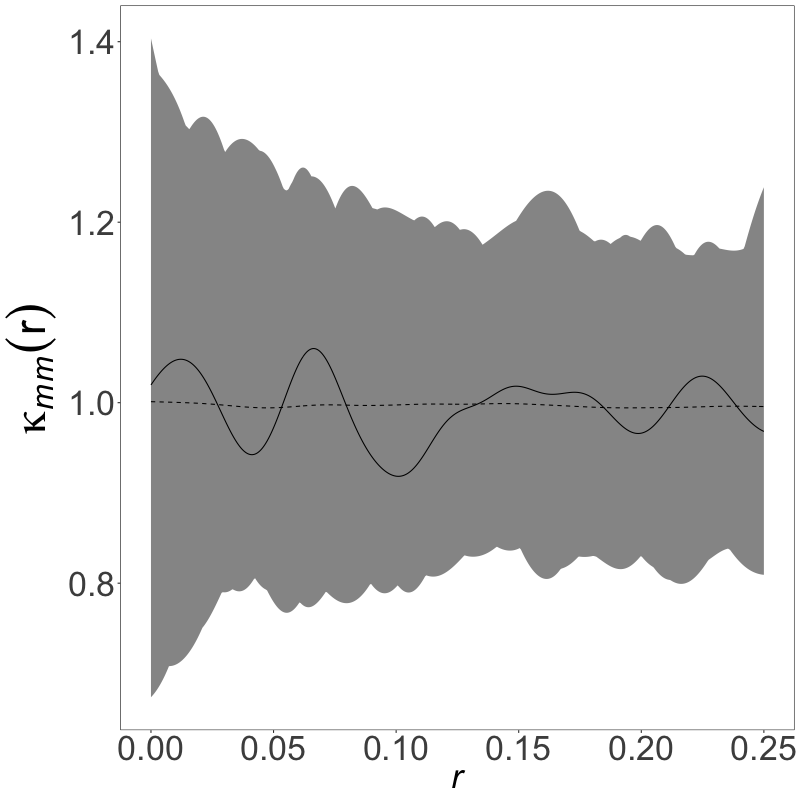}
    \caption{
    One of the 100 simulated point patterns, together with its corresponding global envelope tests based on $\kappa_{mm}^{\textrm{inhom}}$ and $\kappa_{mm}$.
    }
    \label{fig:assoinhomlgcp}
\end{figure}

\subsection{Mark variation}

To study the performance of our proposed inhomogeneous mark variogram $\gamma^{\mathrm{inhom}}_{mm}$, we introduce variation in the mark distribution by incorporating uniform noise. Specifically, we consider the mark distribution $m(u)=m(x,y)= a(x,y) \sin(\sqrt{x^2+y^2})$ where $a(x,y) \sim U(0, 0.5)$. This construction is intended to create heterogeneity in the mark distribution without enforcing a strong spatial association between nearby marks. As a result, the marks tend to vary in amplitude while roughly following a shared underlying trend, leading to a setting where mark variation is present, though not necessarily accompanied by mark correlation. Given the considered mark distribution and the spatial distribution of points, on average, a moderate negative mark variation is expected for a moderate range of spatial distances.

\subsubsection{Inhomogeneous Poisson point process}
We consider an inhomogeneous Poisson point process,
with an intensity function given by $\lambda(u)=\lambda(x,y)= 40(x+y+0.5)^4$. In this scenario, we find that the power of the test for $\gamma_{mm}^{\textrm{inhom}}$ is $94\%$, effectively detecting the imposed mark variation despite the strong underlying spatial intensity gradient. In contrast, the homogeneous version $\gamma_{mm}$, which does not account for spatial inhomogeneity, only reaches a power of $17\%$. This considerable gap underscores the importance of incorporating intensity correction when testing for mark variation, as failure to adjust for spatial inhomogeneity can obscure real patterns and lead to underestimation of mark structure. Figure \ref{fig:varioinhomPoiss} showcases the results for one of the 100 simulated point patterns. Looking at the pattern of the points, and comparing it with the previous cases in Section \ref{sec:asso}, one can see that here there exists a considerable local variation among the marks of nearby points, imposing mark variation without strong association. The inhomogeneous mark variogram $\gamma_{mm}^{\textrm{inhom}}$ correctly detects the imposed moderate variability in the marks by falling outside the envelope under random labelling from the lower bound, whereas the homogeneous mark variogram $\gamma_{mm}$ stays inside the envelope for the whole considered range of spatial distances, failing to identify any significant mark variation. With respect to the type I error probabilities,  $\gamma_{mm}^{\textrm{inhom}}$ and $\gamma_{mm}$ mistakenly detect some significant mark variation for $7\%$ and $4\%$ of the patterns, respectively.

\begin{figure}[!h]
    \centering
    \includegraphics[scale=0.19]{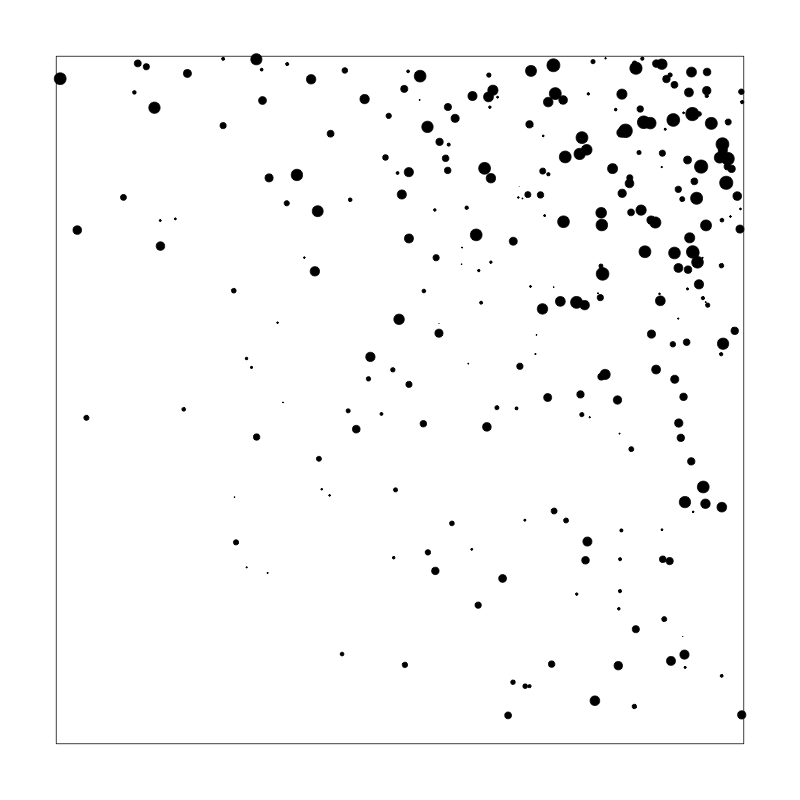}
    \includegraphics[scale=0.18]{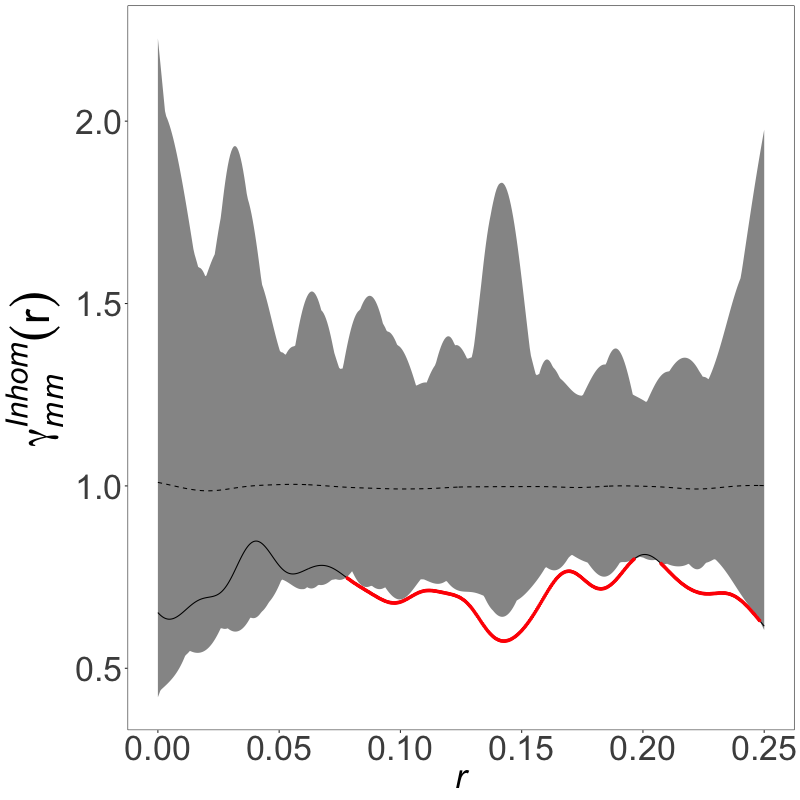}
    \includegraphics[scale=0.18]{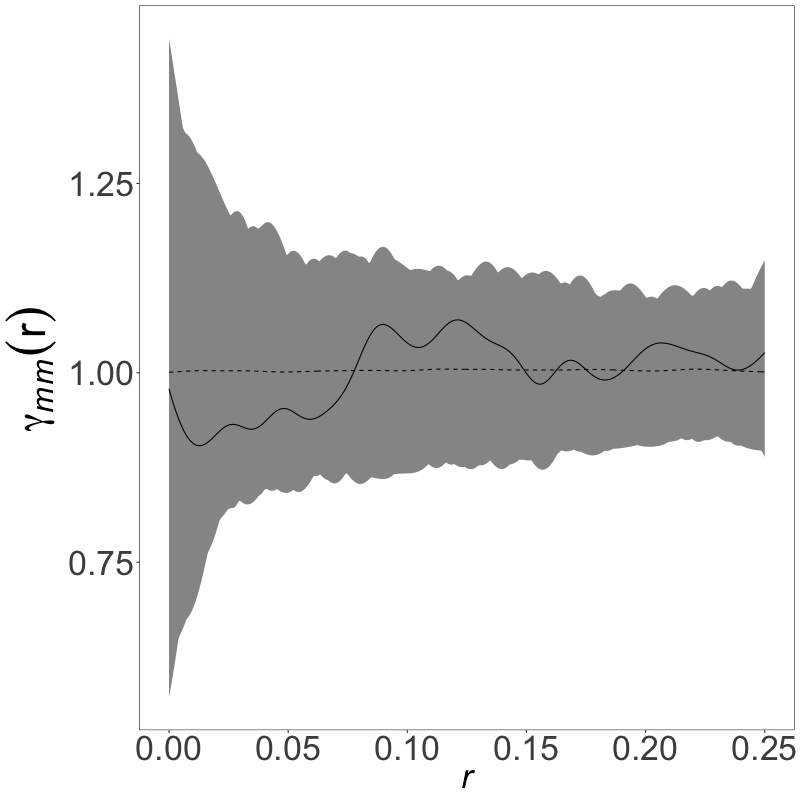}
    \caption{
    One of the 100 simulated point patterns, together with its corresponding global envelope tests based on $\gamma_{mm}^{\textrm{inhom}}$ and $\gamma_{mm}$.
    }
    \label{fig:varioinhomPoiss}
\end{figure}

\subsubsection{Log-Gaussian Cox point processes}

Now, we consider a Log-Gaussian Cox point process where the driving Gaussian random field has a mean function $(x,y) \mapsto \log(200 (x +y + 0.1))$, and a Gaussian covariance function $((x_1,y_1),(x_2,y_2)) \mapsto \exp (- 100 (|| (x_1,y_1) - (x_2,y_2)||)^2 )$, meaning that the intensity function is given as $\lambda(u)=\lambda(x,y)= 200 (x +y + 0.1)\exp (0.5)$. This construction generates spatial point patterns with both strong inhomogeneity and spatial clustering due to the stochastic variation of the intensity surface, combined with a moderate mark variation imposed. Here, we find that  $\gamma_{mm}^{\textrm{inhom}}$ demonstrates a high sensitivity to mark variation, correctly identifying it in  $98\%$ of the simulated patterns, whereas $\gamma_{mm}$ detects the imposed mark variation in only $56\%$ of the cases. 
Figure \ref{fig:variolgcp} showcases the results for one of the patterns in this scenario, where both the spatial heterogeneity and mark variation are evident. One can see that $\gamma_{mm}^{\textrm{inhom}}$ for the entire range of spatial distances correctly identifies the imposed mark variation, while $\gamma_{mm}$ falls within the envelope under random labelling for the entire range of spatial distance, being unable to detect any significant mark variation. In addition, it is noticed that $\gamma_{mm}^{\textrm{inhom}}$ and $\gamma_{mm}$ do not necessarily have similar fluctuations across spatial distances; this is due to the fact that $\gamma_{mm}$ does not account for spatial distribution. Finally, turning to the type I error probability, we observe that both $\gamma_{mm}^{\textrm{inhom}}$ and $\gamma_{mm}$ exhibit a Type I error probability of $0.07$.

\begin{figure}[!h]
    \centering
    \includegraphics[scale=0.19]{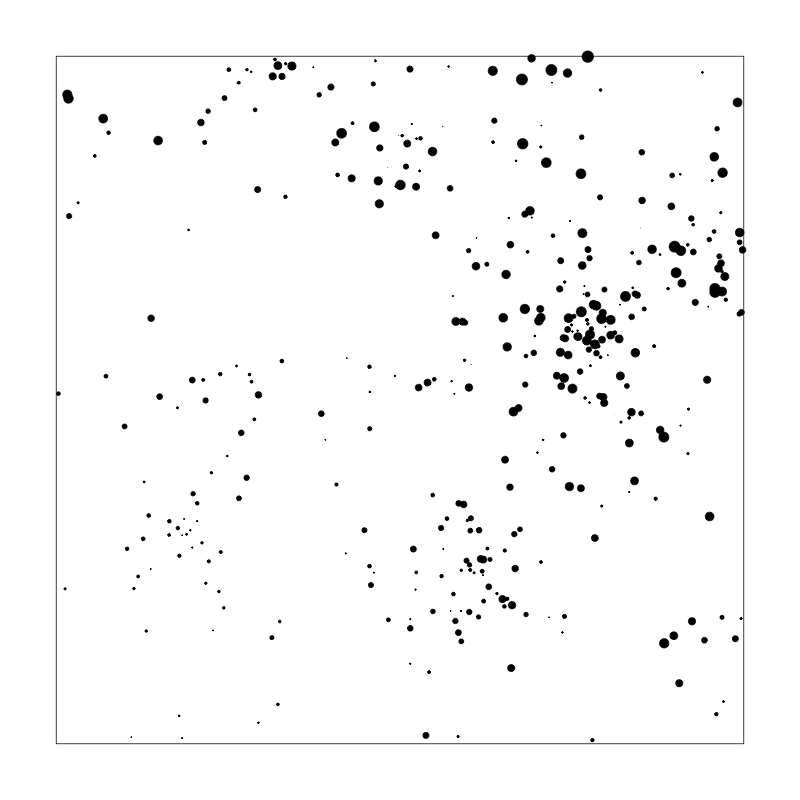}
    \includegraphics[scale=0.18]{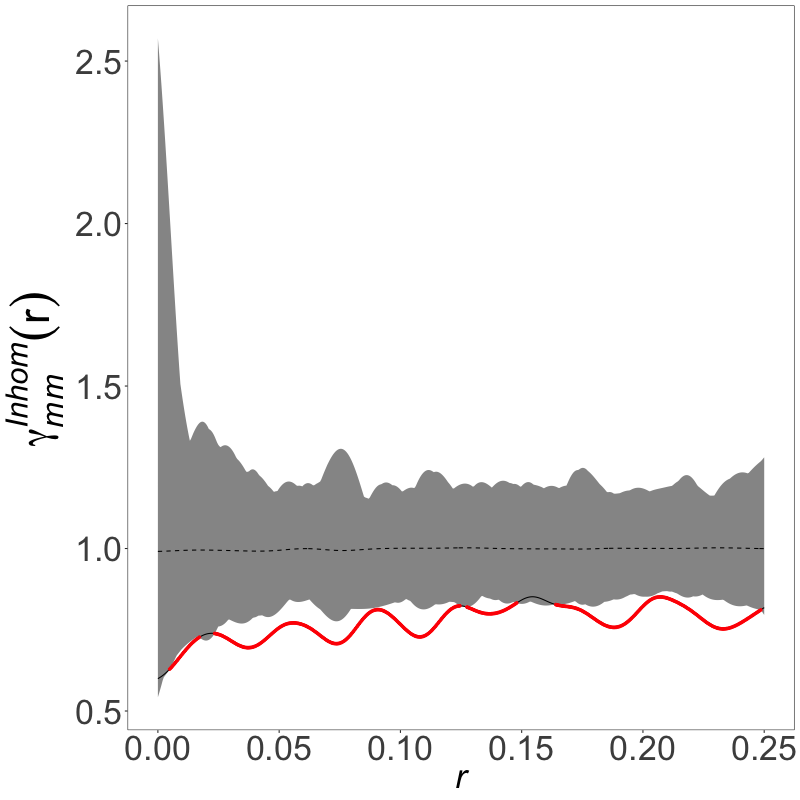}
    \includegraphics[scale=0.18]{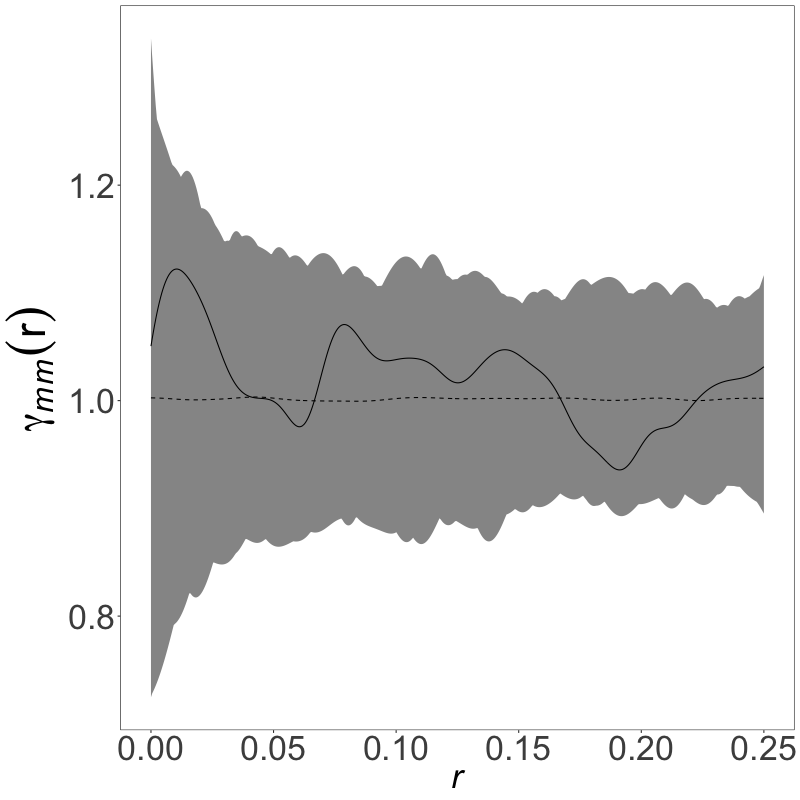}
    \caption{
    One of the 100 simulated point patterns, together with its corresponding global envelope tests based on $\gamma_{mm}^{\textrm{inhom}}$ and $\gamma_{mm}$.
    }
    \label{fig:variolgcp}
\end{figure}

\section{Applications}\label{sec:apps}

This section is devoted to demonstrating the performance of our proposed inhomogeneous mark correlation functions in two real-world scenarios within forestry, where locations of trees are labelled by real-valued marks. In both cases, the observed point patterns show heterogeneity in the distribution of points as well as marks, and we employ $\kappa^{\mathrm{inhom}}_{mm}$ and $\gamma^{\mathrm{inhom}}_{mm}$, in combination with completely non-parametric rank envelope test based on extreme rank lengths with 1000 permutations \citep{myllymaki2017global}, to study potential mark association/variation among trees. The spatial intensities of the trees are estimated using the kernel-based intensity estimator \eqref{e:kde.2D.JD} in combination with Cronie and van Lieshout's criterion for choosing the smoothing bandwidth parameter \citep{cronie2018non}.

\subsection{Longleaf pine trees}

Longleaf data, available through the \textsf{R} package \textsf{spatstat.data}, contain the locations and diameters at breast height for 584 Longleaf pine trees (Pinus palustris) in a $200 \times 200$ meter region in southern Georgia (USA); data were collected by \cite{platt1988population}. The diameters range from $2$ to $75.90$ centimetres, with an average of $26.84$ and a variance of $336.03$. 
Figure \ref{fig:longleafdata} displays the locations of Longleaf pine trees, with their diameters at breast height (dbh) represented by the size of the points rather than numeric labels for clarity. The majority of the trees are adult, having $\text{dbh}>30$, whereas some younger trees with $\text{dbh}<30$ are located with a tendency towards the north-east of the forest; these differences in the spatial distribution of dbh highlight the inhomogeneity of marks. In addition, in Figure \ref{fig:longleafdata}, we show the estimated spatial intensity of the trees, in which a clear inhomogeneity in the spatial distribution of the trees is evident. We additionally show the spatially varying average and variance of tree dbh, using the Nadaraya-Watson smoother \citep[Chapter 6]{Baddeley2015}, both of which highlight the existence of inhomogeneity in the distribution of the marks. These existing inhomogeneity in the distribution of trees, as well as their dbh, highlights the need for mark correlation functions that account for such heterogeneity.

\begin{figure}
    \centering
    \includegraphics[scale=0.2]{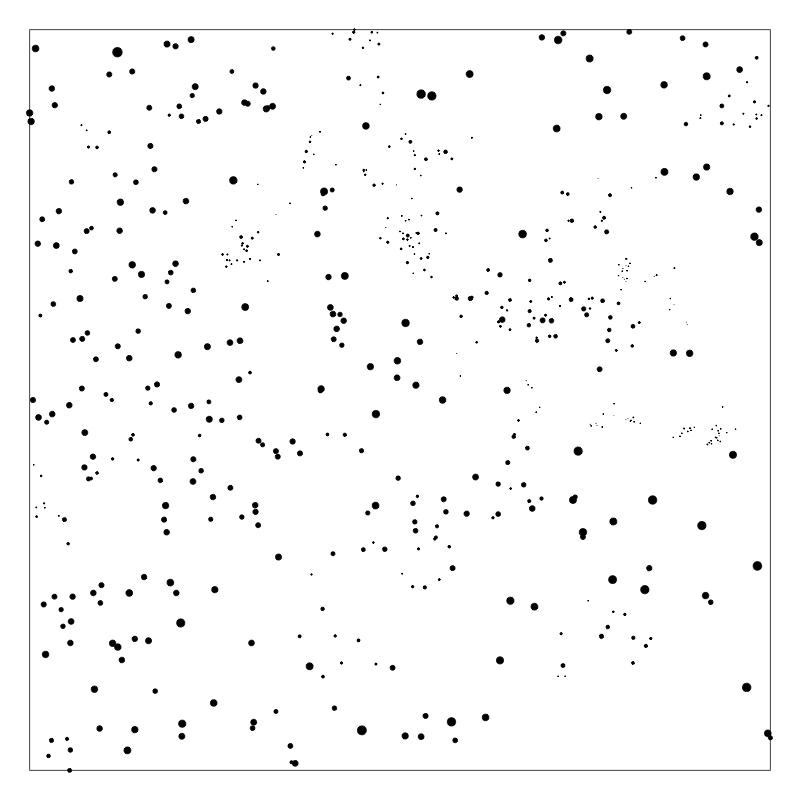}
    \quad
    \includegraphics[scale=0.2]{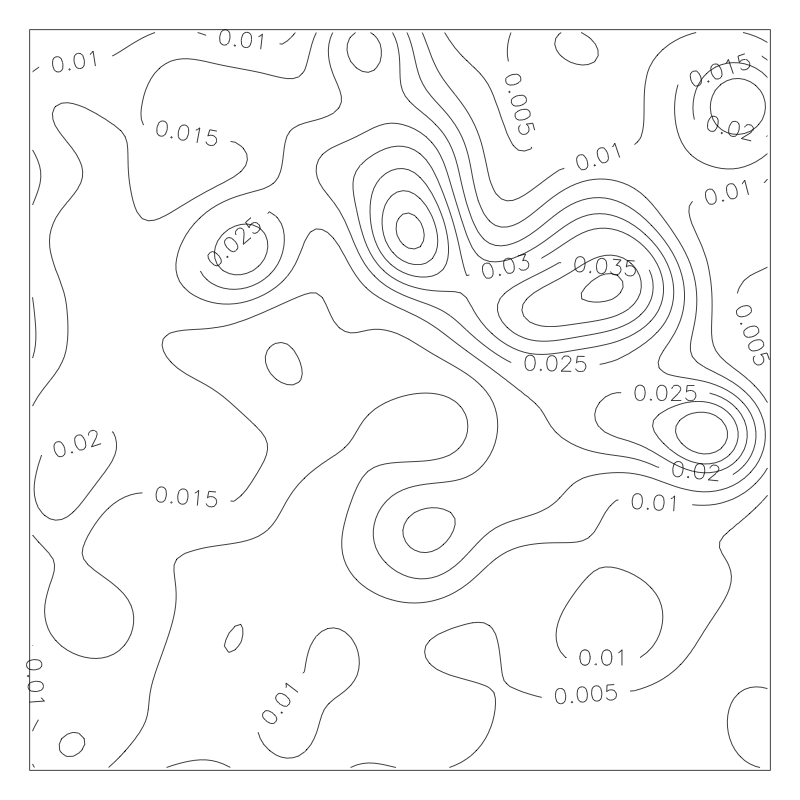}
    \includegraphics[scale=0.2]{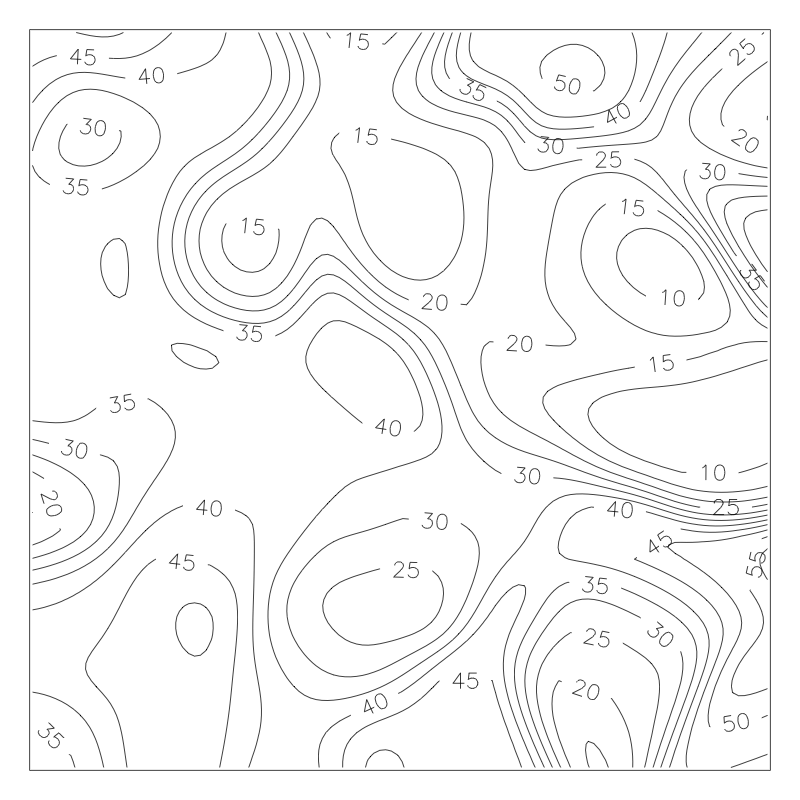}
    \quad
    \includegraphics[scale=0.2]{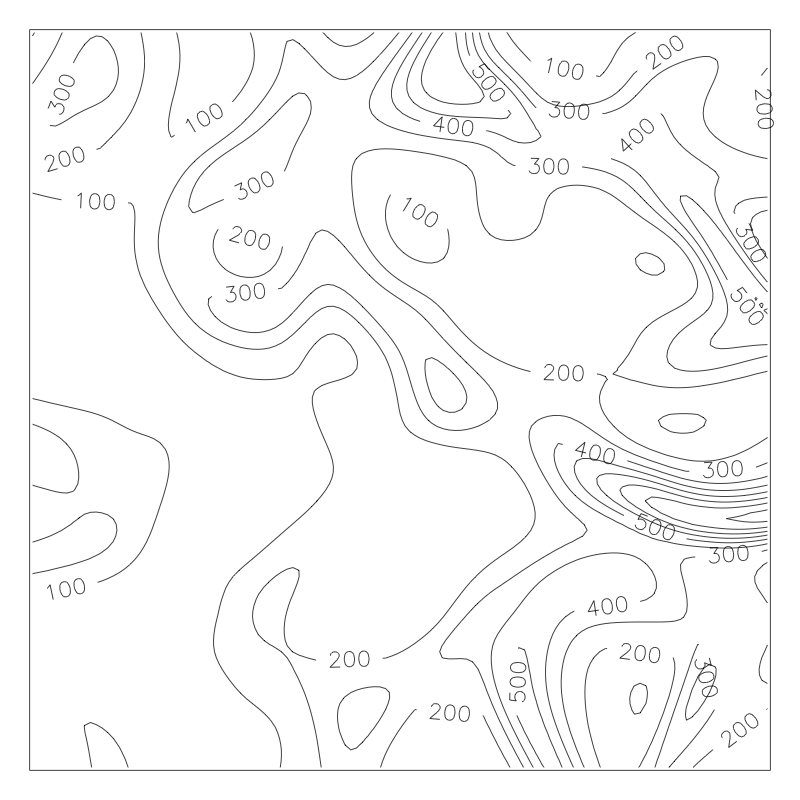}
    \caption{
    From left to right: location and diameter at breast height of Longleaf pine trees (the larger the point, the larger the diameter), estimated spatial intensity, spatially varying average tree diameter in centimetres, and spatially varying variance of tree diameter in centimetres squared.
    }
    \label{fig:longleafdata}
\end{figure}

To address this identified need, we proceed by calculating the inhomogeneous mark correlation functions $\kappa^{\mathrm{inhom}}_{mm}$ and $\gamma^{\mathrm{inhom}}_{mm}$ to detect potential associations and variations among marks while accounting for the spatial inhomogeneity of the point pattern. For comparison, we also compute the homogeneous versions of these mark correlation functions, denoted as $\kappa_{mm}$ and $\gamma_{mm}$. The results, along with global envelope tests based on 1000 permutations, are presented in Figure \ref{fig:longleafresults}. The homogenous mark variogram $\gamma_{mm}$ stays below the obtained envelope under random labelling for almost the entire considered spatial range of distance, identifying significantly lower variation among trees' dbh than the expected variation under random labelling. Additionally, it highlights that such a significant difference is quite stronger for smaller spatial distances, as it has values close to zero, indicating similar dbh for very close trees. The inhomogeneous mark variogram $\gamma^{\mathrm{inhom}}_{mm}$, however, offers slightly different conclusions. It stays below the obtained envelope only for pairs of trees whose spatial distance is less than approximately 25 meters. Therefore, no significant variation in the spatial distribution of trees' dbh is detected beyond the pairwise spatial distance of 25 meters. In other terms, $\gamma^{\mathrm{inhom}}_{mm}$ suggests that in high-intensity areas of the forest, the dbh variation is slightly higher than that detected by $\gamma_{mm}$; this variation might reflect competition among trees for resources such as light and nutrients which lead to individual trees grow at different rates depending on their micro-environment. In contrast, in areas with lower intensity, dbh variation is not identified as significant, which might be explained by the fact that in such areas the competition is supposed to be weaker among trees, and thus they may grow more uniformly.
Turning to the homogeneous and inhomogeneous mark correlation functions $\kappa_{mm}$ and $\kappa^{\mathrm{inhom}}_{mm}$ to study potential mark association among trees, interestingly, we observe that they show completely different outcomes. From the outcomes based on $\kappa_{mm}$, a strong negative association is detected for pairs of trees with a spatial distance less than approximately 15 meters; for larger pairwise distances, it almost stays within the envelope, indicating no significant association. In contrast, the inhomogeneous mark correlation function $\kappa^{\mathrm{inhom}}_{mm}$ shows strong positive association among trees for the entire range of pairwise spatial distances, as it stays above the envelope. This suggests that, when considering any pair of trees and adjusting for spatial intensity, there is a tendency for at least one tree to have a large dbh. The findings based on $\gamma^{\mathrm{inhom}}_{mm}$ and $\kappa^{\mathrm{inhom}}_{mm}$ align more closely with the spatially varying mean and variance of trees' dbh shown in Figure \ref{fig:longleafdata}.

\begin{figure}
    \centering
    \includegraphics[scale=0.2]{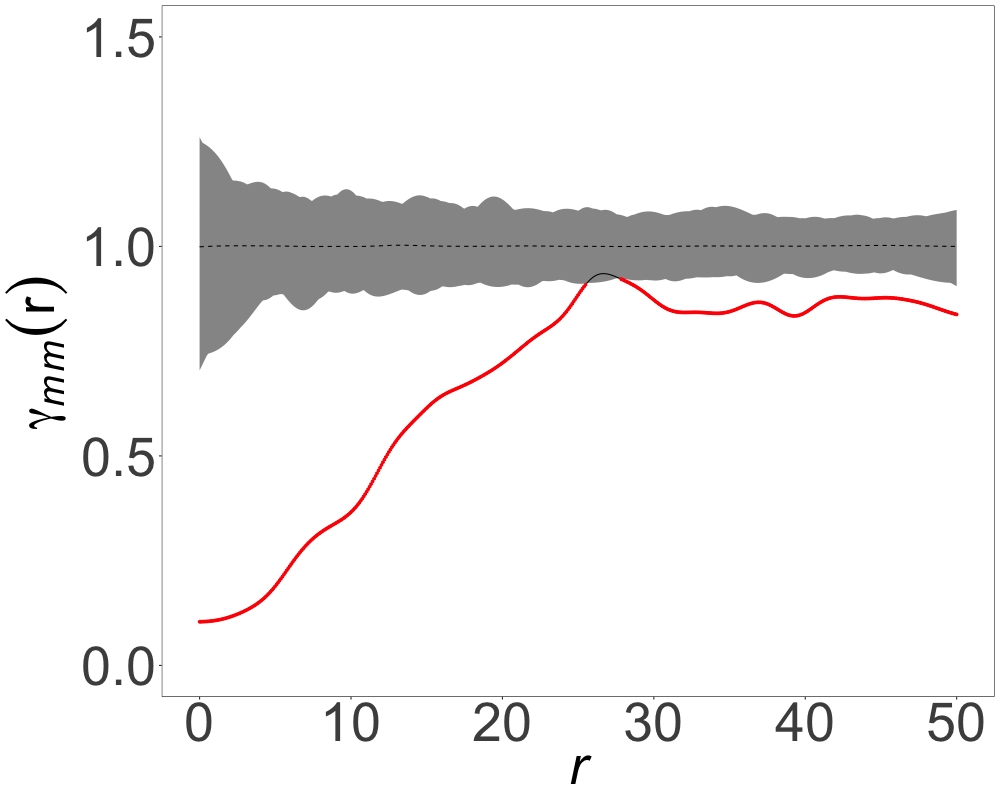}
    \includegraphics[scale=0.2]{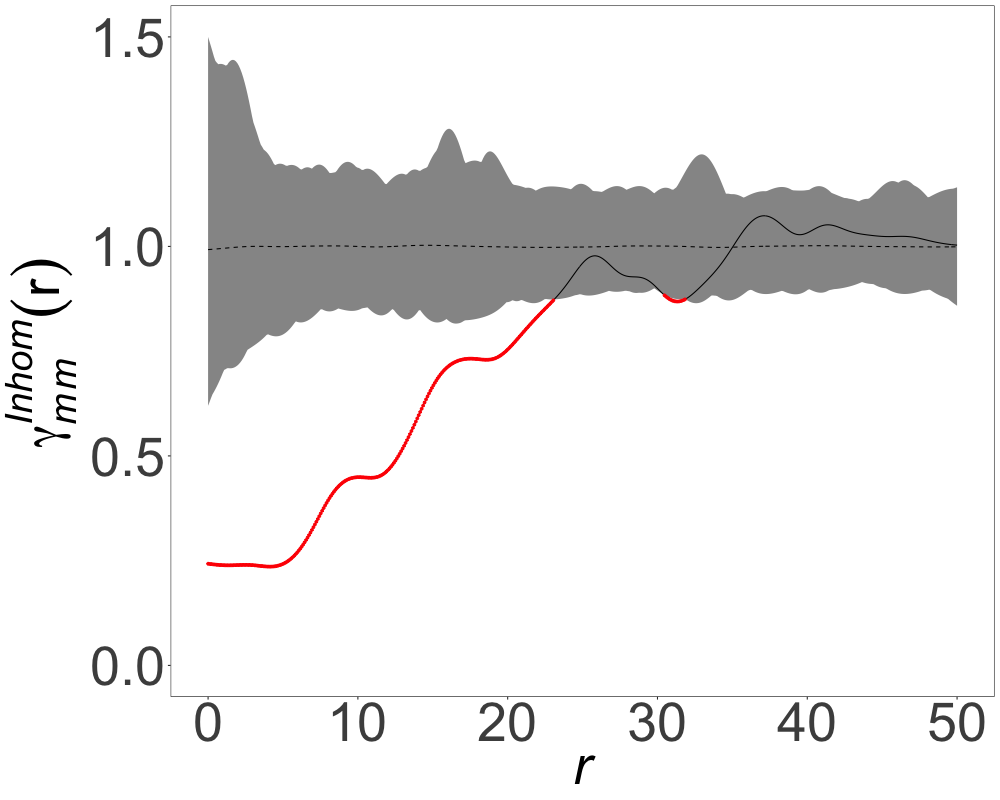}
    \includegraphics[scale=0.2]{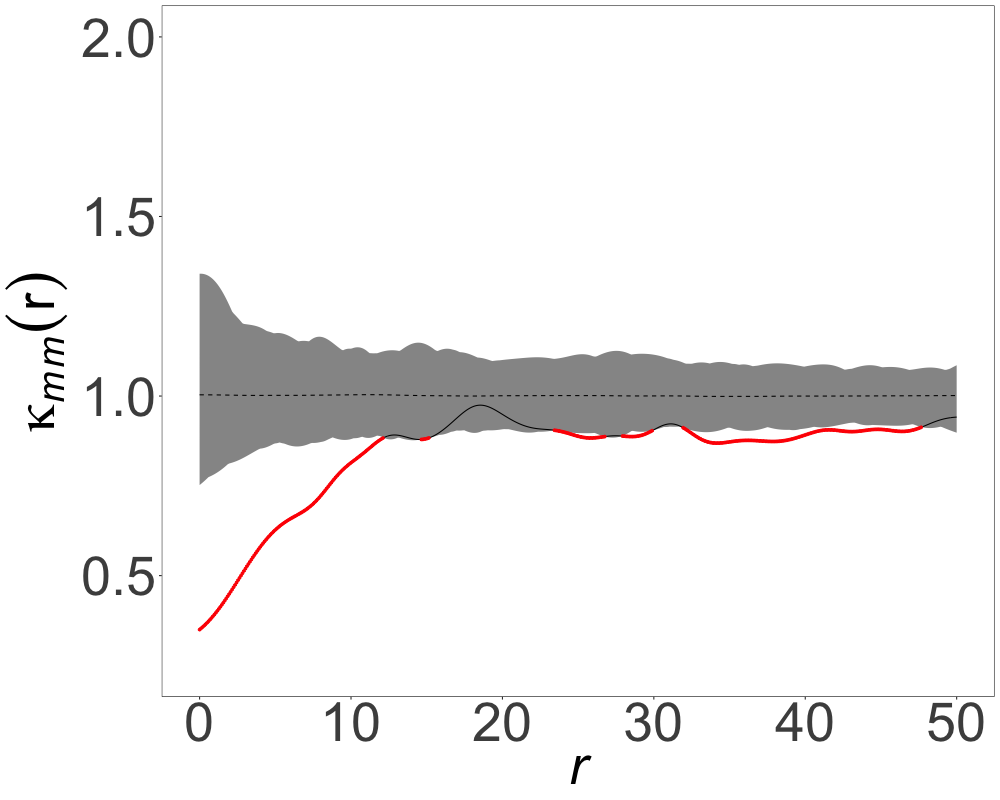}
    \includegraphics[scale=0.2]{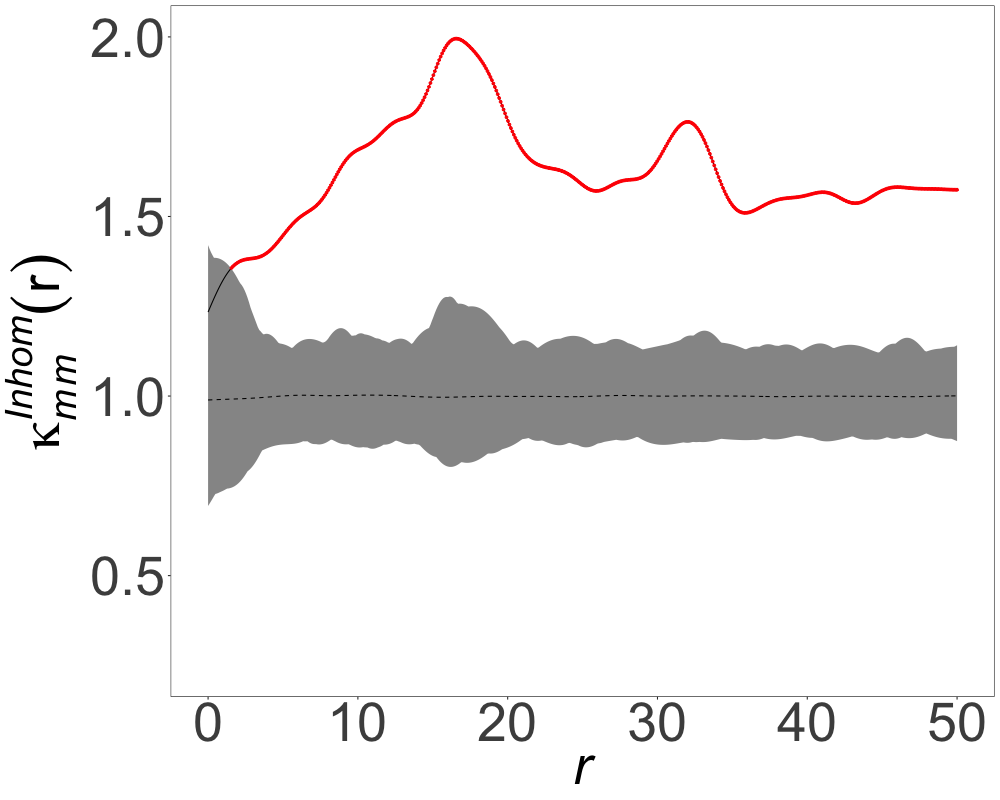}
    \caption{
    From top-left to bottom-right: homogeneous mark variogram $\gamma_{mm}$, inhomogeneous mark variogram $\gamma_{mm}^{\mathrm{inhom}}$, homogeneous mark correlation function $\kappa_{mm}$, inhomogeneous mark correlation function $\kappa_{mm}^{\mathrm{inhom}}$.
    }
    \label{fig:longleafresults}
\end{figure}

\subsection{Pfynwald}

In this section, we study the location and height of trees from the \textit{Pfynwald} project, available via an Open Database License\footnote{\url{https://opendata.swiss}}, which originates from a long-term irrigation experiment begun in 2003 within Switzerland's \textit{Pfyn-Finges} national park \citep{pfynwald:2016}. 
Here, we are interested in investigating the spatial distribution of trees in relation to their height, aiming to identify potential patterns of variation and/or association. As the height measurements are only available for a subset of trees, our analysis is limited to the measurements collected in 2009. The dataset contains the locations and height of $289$ trees shown in the right-hand side of Figure \ref{fig:Pfynwalddata}, where one can see a higher concentration of trees in the south of the forest. The height of trees varies between $1.1$ to $14.1$ meters, with an average of $9.516$ and a variance of $20.03$; $75\%$ of trees are taller than $11$ meters. Looking at Figure \ref{fig:Pfynwalddata}, we can see that most short/young trees are located in the centre and south of the forest, while taller trees are spread around the forest. This combination of short and tall trees highlights a variation in the spatial distribution of the trees' height. Similar to the case of Longleaf pines, here we first estimate the spatial intensity of the trees as well as the spatially varying average and variance of the trees' height, using the Nadaraya-Watson smoother. From Figure \ref{fig:Pfynwalddata}, the estimated spatial intensity of trees is highest in the southern part of the forest, decreases toward the central region, and shows a slight increase again in the north, highlighting the inhomogeneity of the pattern. Looking at the varying average and variance of the trees' height, we can further observe that the average height is highest in the north, lowers down towards the centre, and slightly increases in the south. In contrast, the estimated spatially varying variance of the trees' height is highest in the central region, decreasing toward both the north and south, with a slight increase in the far south-west and south-east.

\begin{figure}
    \centering
    \includegraphics[scale=0.27]{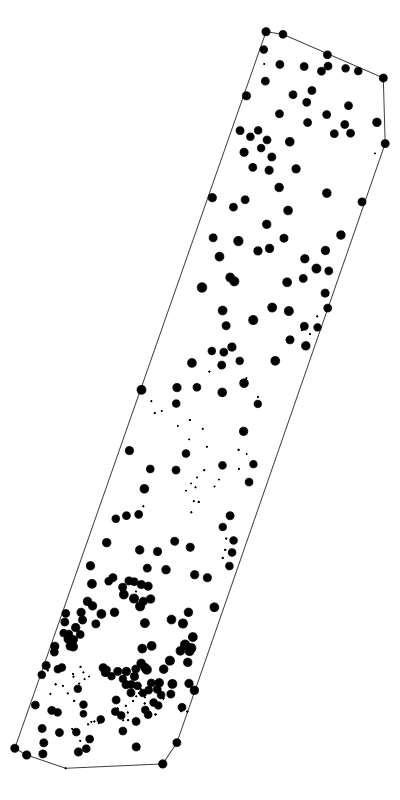}
    \quad
    \includegraphics[scale=0.27]{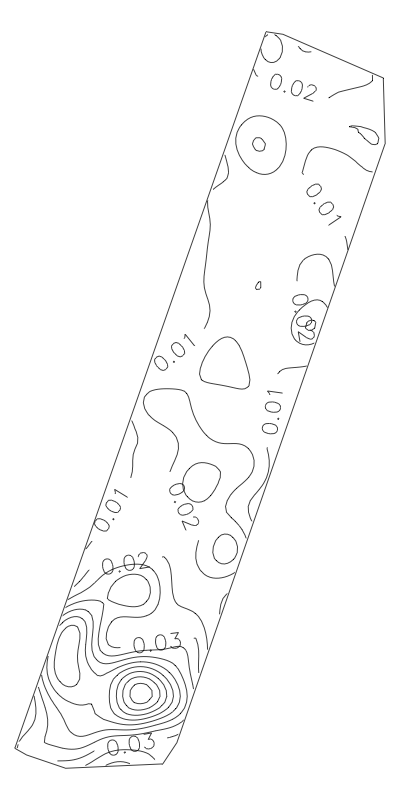}
    \includegraphics[scale=0.27]{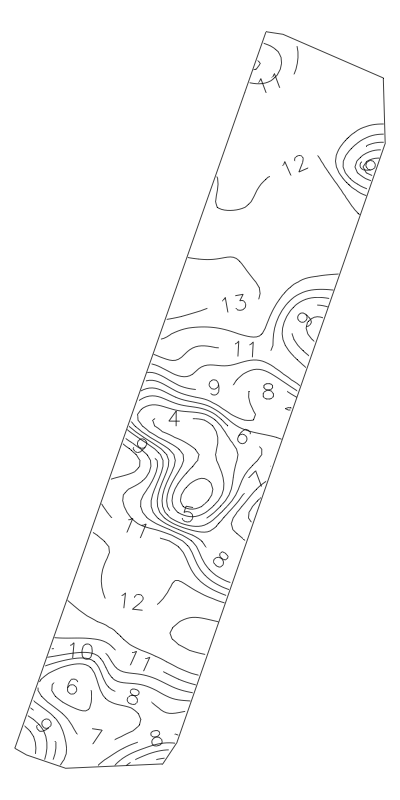}
    \quad
    \includegraphics[scale=0.27]{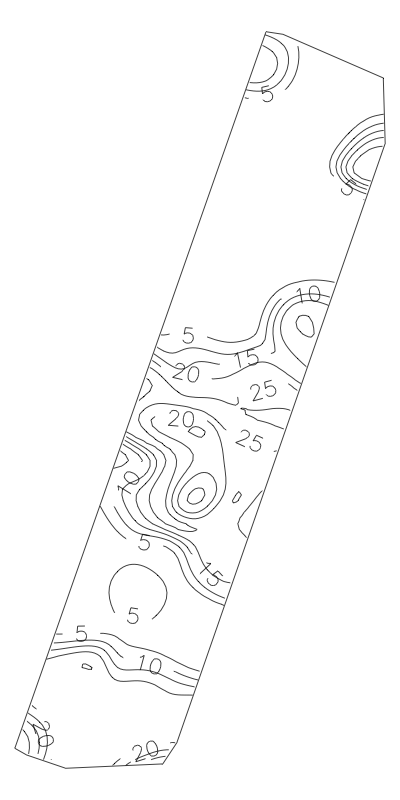}
    \caption{
    From left to right: location and height of trees (the larger the point, the taller the tree), estimated spatial intensity, spatially varying average tree height in centimetres, and spatially varying variance of tree height in centimetres squared.
    }
    \label{fig:Pfynwalddata}
\end{figure}

Next, we make use of our proposed inhomogeneous mark correlation functions $\kappa_{mm}^{\mathrm{inhom}}$ and $\gamma_{mm}^{\mathrm{inhom}}$ jointly with their homogeneous versions, denoted as $\kappa_{mm}$ and $\gamma_{mm}$. Results are presented in Figure \ref{fig:Pfynresults}. Interestingly, we find that the homogeneous mark correlation functions fail to detect any significant association/variation as they both stay within their corresponding envelopes under random labelling for the entire range of spatial distances. For very small spatial distances, both functions remain near the lower bounds of their respective envelopes, indicating a tendency toward negative mark association and variation among nearby trees, although this does not reach statistical significance. Turning to our proposed $\kappa_{mm}^{\mathrm{inhom}}$ and $\gamma_{mm}^{\mathrm{inhom}}$, apart from very small distances, we find that $\kappa_{mm}^{\mathrm{inhom}}$ consistently lies above the envelope generated under random labelling, indicating positive mark association among the height of the trees. This suggests that for any two trees separated by at least 7 meters, it is likely that at least one of them is tall. The fact that $\kappa_{mm}^{\mathrm{inhom}}$ stays within the envelope for spatial distances less than 7 meters means that in high-intensity areas within the Pfyn forest, the height of the trees seems to behave randomly compared to the same for the trees separated by larger distances. This, in turn, suggests that in high-intensity areas, trees compete for resources, and thus, leading to heterogeneous growth conditions and consequently greater variability in the height of trees among closely spaced individuals. Moreover, $\gamma_{mm}^{\mathrm{inhom}}$ falls below its corresponding envelope for pairs of trees separated by distances between 7 and 28 meters approximately, implying that the variation in the trees' height within such pairs is lower than would be expected under random labelling; under random labelling we expect the variation among the trees' height to be at least slightly higher. In other words, pairs of trees with a spatial distance between 7 and 28 meters exhibit a degree of similarity in height. This is in agreement with our observations based on $\kappa_{mm}^{\mathrm{inhom}}$.


\begin{figure}
    \centering
    \includegraphics[scale=0.2]{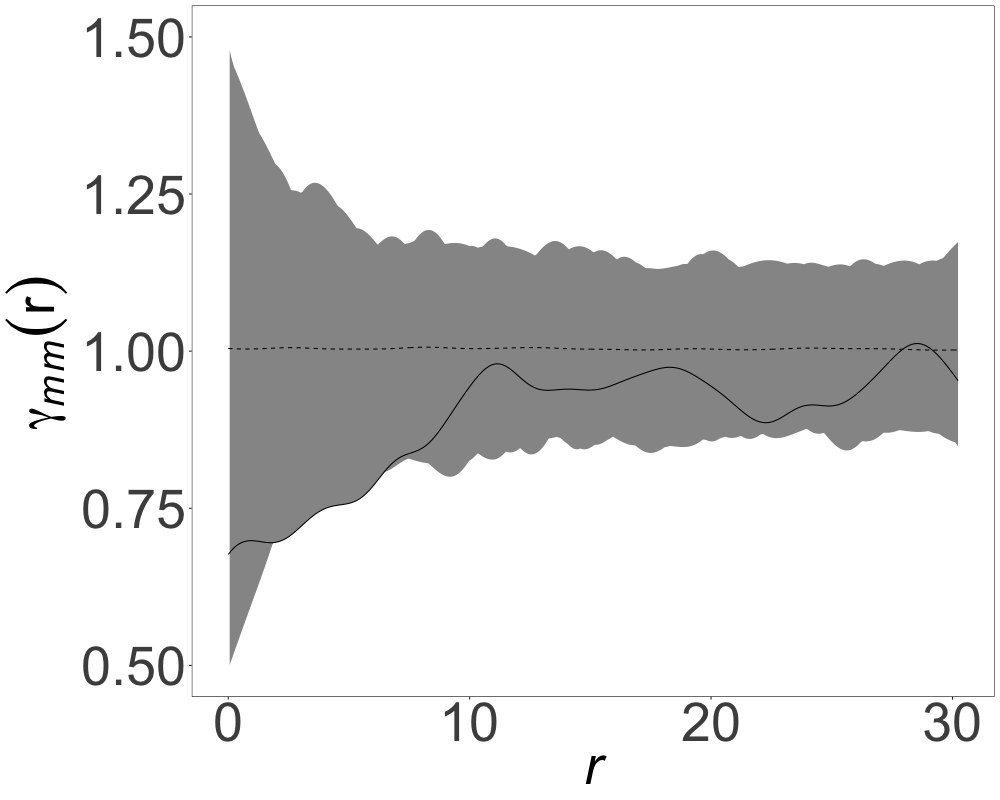}
    \includegraphics[scale=0.2]{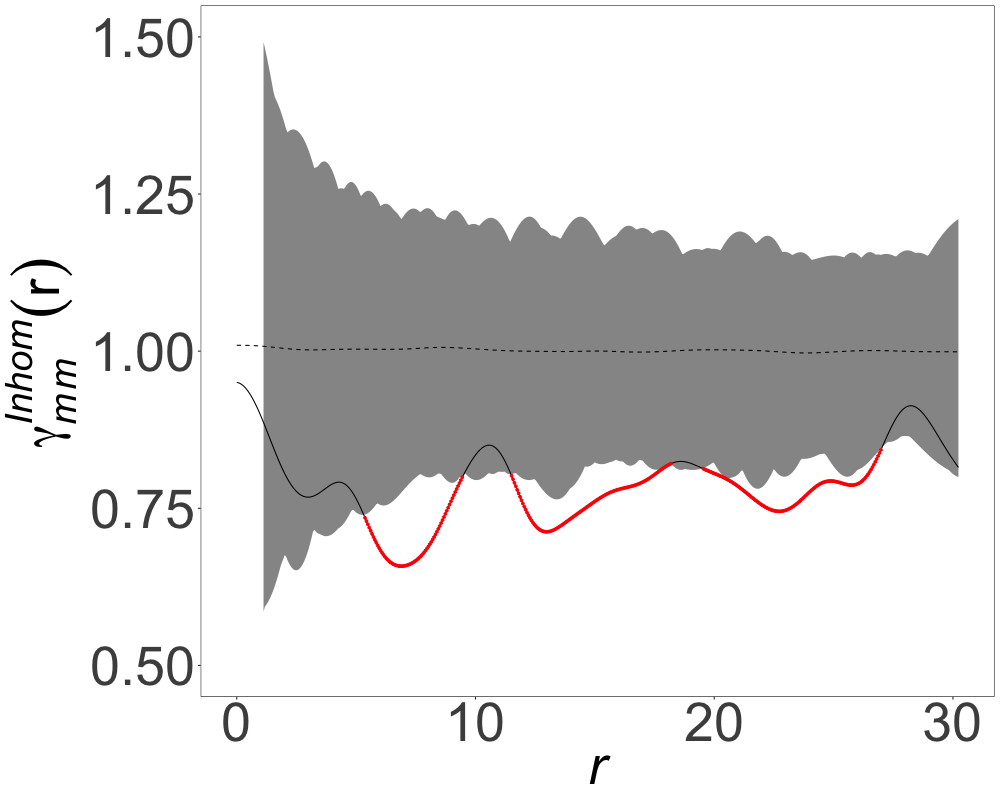}
    \includegraphics[scale=0.2]{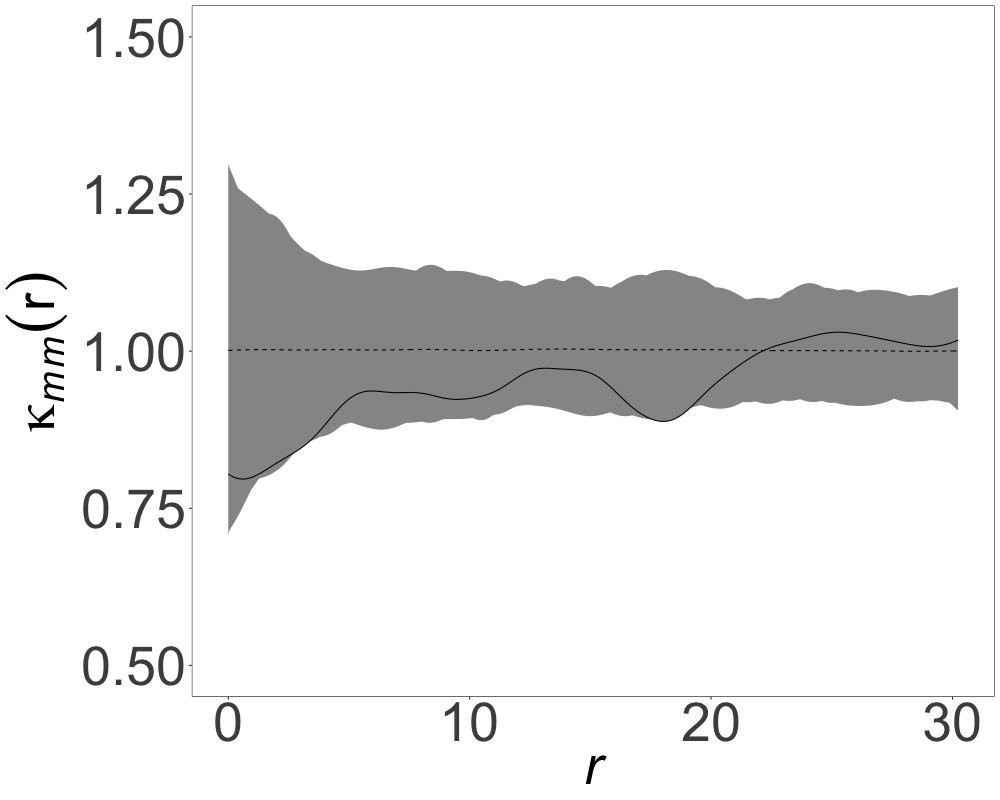}
    \includegraphics[scale=0.2]{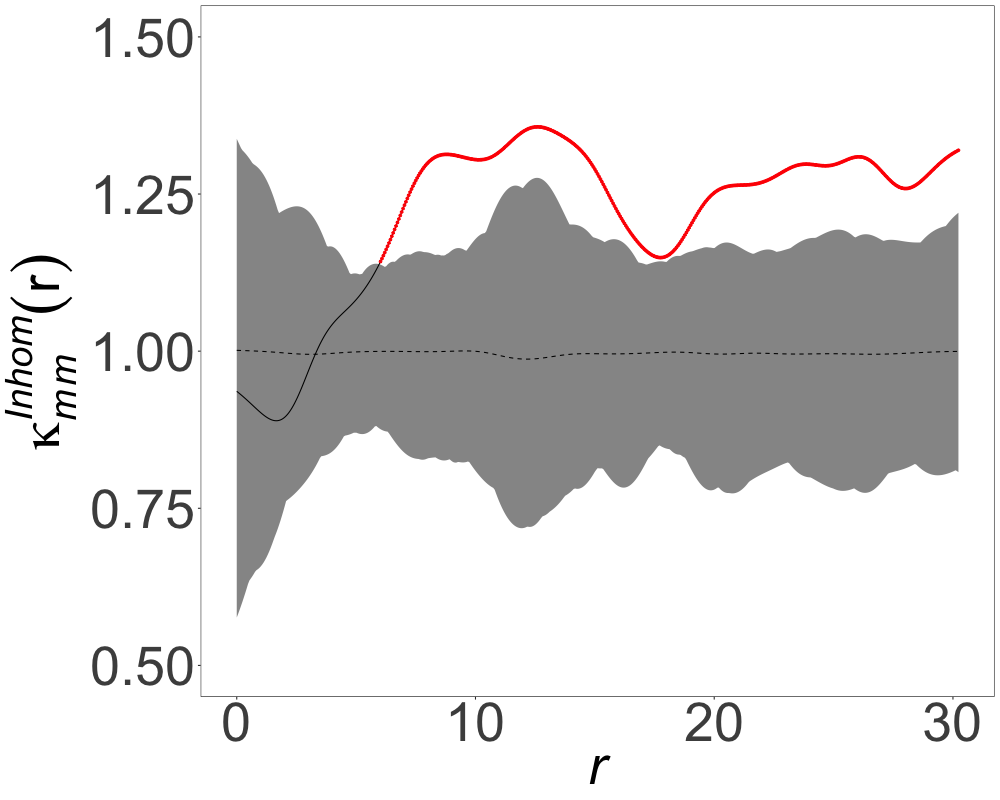}
    \caption{
    From top-left to bottom-right: homogeneous mark variogram $\gamma_{mm}$, inhomogeneous mark variogram $\gamma_{mm}^{\mathrm{inhom}}$, homogeneous mark correlation function $\kappa_{mm}$, inhomogeneous mark correlation function $\kappa_{mm}^{\mathrm{inhom}}$.
    }
    \label{fig:Pfynresults}
\end{figure}

\section{Discussion}\label{sec:diss}

In numerous applications of marked point processes, where heterogeneity in the distribution of points is evident, traditional mark correlation functions, relying on stationarity assumptions, fail to accurately capture the actual underlying distributional behaviour of the marks. In particular, they might not be efficient at detecting spatially dependent associations or variations among the marks, as they assume that all points have the same number of neighbours at a given interpoint distance. This key limitation arises from their disregard for spatial variations in the intensity function.
Motivated by two applications in forestry, where the spatial distributions of trees are non-uniform, we have introduced the class of inhomogeneous mark correlation functions for general marked point processes where marks are real-valued. We have shown that unnormalised mark correlation functions can be expressed based on the ratio of two pair correlation functions.
This paved the way for proposing inhomogeneous mark correlation functions, which turned out to be quite effective in situations where the spatial distribution of points exhibits clear heterogeneity. Supported by the statistical properties of the estimators for the inhomogeneous pair correlation function, we have proposed nonparametric estimators for our proposed inhomogeneous mark correlation functions that remain ratio-unbiased. In particular, we found that the homogeneous Stoyan’s mark correlation function yields highly misleading results when applied to inhomogeneous point patterns. These inaccuracies go beyond merely detecting the presence or absence of mark association, as they can also lead to incorrect conclusions about the positivity/negativity of the actual association. In contrast, we found that the homogeneous mark variogram is less sensitive to heterogeneity in the point distribution. Its limitations primarily lie in detecting potential mark variations; when such variations are detected, the positivity/negativity of mark variations is generally correct, although the strength of the variation is often overestimated. In other words, our inhomogeneous mark variogram outperforms its homogeneous counterpart in accurately capturing the strength of mark variation.

The inhomogeneous mark correlation functions, applicable to marked point processes on general state spaces, extend to various mark settings such as composition/function/graph-valued cases as specific types of object-valued marks, in addition to real-valued marks, without being confined to planar point patterns. Notably, for analysing marked point patterns on linear networks where points are often unevenly distributed across the network, this framework offers a valuable alternative to current methodologies that do not account for heterogeneity in the distribution of points. In addition, based on the same formulation, one can define inhomogeneous local indicators of mark associations/variations, and further extend the framework to cross-characteristics and inhomogeneous multivariate point processes with additional quantitative marks.

\vspace{-0.7cm}

\bibliographystyle{biom}
\bibliography{inhom}

\end{document}